\documentclass[10pt, journal]{IEEEtran}
\usepackage{amsmath}
\usepackage{amssymb}
\usepackage{bm}
\usepackage{mdwmath}
\usepackage{graphicx}
\graphicspath{{../Figure/PDF/}{../Biography/}{../Figure/EPS/}}
\usepackage{xcolor}
\usepackage{subfigure}
\usepackage{stfloats}
\usepackage{makecell}
\usepackage{multirow}
\usepackage{acronym}
\usepackage{enumerate}
\usepackage[noadjust]{cite}
\usepackage{url}
\usepackage[ruled, vlined, linesnumbered]{algorithm2e}
\SetKw{Kw}{Initialization:}
\SetKwRepeat{Repeat}{Repeat}{Until}
\SetKwFor{For}{foreach}{do}{endfor for}%
\SetKwFor{for}{for}{do}{endfor for}%
\usepackage{hyperref}
\usepackage{cases}
\usepackage{bbm}
\usepackage{booktabs}
\usepackage{amsthm}
\theoremstyle{definition}
\newtheorem{definition}{\textit{Definition}}
\newtheorem{theorem}{\textbf{\textit{Theorem}}}

\definecolor{lightblue}{RGB}{220,235,255}

\newcommand\highlightReference[1]{%
  \expandafter\newcommand\csname highlightReference-#1\endcsname{}%
}
\let\oldbibitem\bibitem
\def\bibitem#1 #2\par{%
  \expandafter\ifx\csname highlightReference-#1\endcsname\relax
    \oldbibitem{#1}#2\par
  \else
    \oldbibitem{#1}\highlight{#2}\par
  \fi
}

\newcommand\highlight[1]{\textcolor{blue}{#1}}
\begin{document}
\title{Concurrent Task Orchestration: Hierarchical Trusted Collaborator Group Formation for \\Constrained Resource Allocation}

\author{Jiazhi Chen,~\IEEEmembership{Member,~IEEE,}
		Xianbin Wang,~\IEEEmembership{Fellow,~IEEE,}
  \\and Xuemin (Sherman) Shen,~\IEEEmembership{Fellow,~IEEE}
\thanks{This work was supported in part by the Natural Sciences and Engineering Research Council of Canada (NSERC) Discovery Grants Program under Grants RGPIN-2024-05720 and RGPIN2026-05399, and in part by the Canada Research Chairs Program under Grant CRC-2023-00336. \textit{(Corresponding author: Xianbin Wang.)}}
\thanks{Jiazhi~Chen and Xianbin~Wang are with the Department of Electrical and Computer Engineering, Western University, London, ON N6A 5B9, Canada, (e-mails: jche629@uwo.ca, xianbin.wang@uwo.ca).}
\thanks{Xuemin (Sherman) Shen is with the Department of Electrical and Computer Engineering, University of Waterloo, Waterloo, ON N2L 3G1, Canada (e-mail: sshen@uwaterloo.ca).}}
\maketitle

\begin{abstract}
Effective collaboration among resource-constrained mobile intelligent machines, such as Unmanned Aerial Vehicles (UAVs), is essential for the completion of complex tasks. However, concurrent task demands, limited distributed resources, and dynamic collaborator availabilities make it challenging to satisfy task-specific completion goals in diverse collaborative systems. To address these challenges simultaneously, we propose a new framework of trusted collaborator group formation and the supporting resource allocation strategy for concurrent task orchestration. Specifically, a collaborative task is decomposed into subtasks subject to dynamic resource availability and inter-subtask dependencies. A task utility model is developed for collaborator group evaluation support subtask allocation according to task completion goals and unified value-risk analysis. Based on this model, concurrent task orchestration is formulated as a long-term utility maximization problem under $\alpha$-fairness. To solve the problem efficiently, we develop a hierarchical matching strategy that first clusters tasks and collaborators into coarse-grained sets and subsequently performs fine-grained collaborator group formation and task allocation within matched sets. Task-specific trust is predicted through spatial-temporal Graph Neural Networks (GNNs) and embedded into a fairness-aware Deep Reinforcement Learning (DRL) framework to continuously update the matching policies. Simulation results demonstrate that the proposed strategy improves concurrent task completion performance by 22\% over representative baselines under resource constraints.
\end{abstract}
\begin{IEEEkeywords}
Task orchestration, trust, collaborator group, graph neural networks, deep reinforcement learning, fairness. 
\end{IEEEkeywords}

\maketitle

\section{Introduction}
\IEEEPARstart{E}{mpowered} by the rapid convergence of inherent communication and computing technologies, future wireless networks will support a plethora of applications, including smart manufacturing and intelligent transportation systems, which involve increasingly complex collaborative tasks such as multi-robot assembly or vehicle platoon control \cite{11510260}. Amidst the expeditious technological evolution, limited resources available at each device and growing heterogeneity among tasks necessitate intelligent and reliable collaboration among connected machines for effective task completion \cite{11358800}. For instance, a group of connected intelligent machines (e.g., Unmanned Ground Vehicles (UGVs)) with computing resources and communication capabilities, could be selected to handle tasks delegated from other resource-constrained intelligent machines (e.g., Unmanned Aerial Vehicles (UAVs)) \cite{DRLsurvey}. However, selecting suitable collaborator groups is difficult due to intricate requirements for capability and reliability.
\vspace{-0.2cm}
\subsection{Motivation} \label{sectionImotivation}
Given the heterogeneity and complexity of future collaborative systems, achieving efficient concurrent task orchestration through effective collaboration among intelligent machines under stringent resource constraints can be extremely challenging. This challenge primarily stems from the trade-off between meeting diverse requirements of concurrent tasks, possibility of involving different groups of collaborators for one same task, and drastically increased complexity to optimize overall system efficiency under fairness considerations.

From the perspective of an individual task execution, both the task decomposition into subtasks and the selection of a group of collaborators jointly entail task completion value and risks. On the one hand, task completion value is fulfilled through meeting diverse task completion goals such as latency and energy consumption requirements \cite{11193666}, thus imposing intricate criteria on how tasks are decomposed and collaborator groups are selected. Proper task decomposition supports fine-grained collaboration and aligns subtask-level completion goals with available resources. However, when task completion goals are ambiguously defined and trade-offs between competing goals are improperly balanced, the effectiveness of best-effort strategies such as Quality-of-Service (QoS)-based task offloading is significantly reduced \cite{liu2025large}. Moreover, attacks launched by malicious machines, such as illegitimate collaborators \cite{11614586} or eavesdroppers \cite{zhao2025against}, further degrade task completion value. On the other hand, dynamic task requirements and varying collaborator conditions introduce significant uncertainty for collaborator group formation and task allocation, thus increasing risks for task completion \cite{risk1}. While some existing approaches such as heuristic-based task offloading select suitable collaborator groups, they often rely on static and complete system information~\cite{dynamicfairness1,convexfairness2}, making them ineffective in highly dynamic scenarios.

From the perspective of system orchestration, sharing limited resources of all collaborators among concurrent tasks remains a major challenge in guaranteeing the task completion value and tolerable risks for each task. Specifically, diverse task requirements induce increased dynamics in resource requirements, thereby intensifying uncertain contention among concurrent tasks for limited resources. This uncertainty in resource contention is eventually translated into unpredictable task completion outcomes and value. Thus, effective orchestration of concurrent tasks is needed to allocate limited resources across all completing tasks, aiming to maximize the collective task completion value and minimize associated risks simultaneously. Existing strategies typically formulate concurrent task orchestration as an optimization problem aiming to maximize the cumulative utilities of all concurrent tasks~\cite{liwanggraph}. However, such objectives may lead to disproportionate outcomes, where some tasks could dominate allocated resources while others are persistently underserved. While short-term fairness is often unattainable under resource constraints, enforcing \textit{long-term fairness}, i.e., fair task allocation over time, is essential to prevent persistent resource monopolization.

Task-specific trust serves as a goal-achieving mechanism to direct the selection of reliable collaborators aligned with diverse completion goals of tasks~\cite{chengoal2025}. We extend this concept from a single collaborator evaluation to a group evaluation, ensuring collective capability and reliability of collaborator groups in meeting the global requirements of tasks. Machine Learning (ML) supports adaptive collaborator group selection in dynamic environments by learning complex relationships between task requirements and collaborator conditions~\cite{11523038}. Among ML algorithms, Deep Reinforcement Learning (DRL) particularly supports sequential decision-making under complexity and dynamics, making it well-suited for optimizing long-term outcomes~\cite{QoSDRL,eveattack1}. Thus, the integration of task-specific trust and DRL facilitates the automatic formation and matching of reliable collaborator groups aligned with evolving task goals. However, while DRL performs well under fully observable environments, partial observability of system states complicates such integration. Moreover, DRL typically maximizes long-term rewards, i.e., the cumulative utilities of all tasks \cite{Sarsa}, without considering long-term fairness, leading to imbalanced resource allocation among tasks. In addition, DRL model training in large-scale systems often becomes time-consuming and suffers from slow convergence~\cite{mao2024dl,DRLsurvey}.

Based on these observations, an effective concurrent task orchestration strategy should jointly account for heterogeneous task completion goals, dynamic collaborator conditions, scalable group formation, and long-term fairness, rather than optimizing these aspects in isolation
\vspace{-0.1cm}
\subsection{Related Work}
Reliable task orchestration in dynamic collaborative systems requires trust-aware collaborator selection and fair resource allocation. Accordingly, we review trust-based, learning-based, and fairness-aware task orchestration.

\textit{1) Trust-Based Task Orchestration:}
Trust has become increasingly important in orchestrating tasks, especially when security threats are involved. Beyond using trust scores derived from historical task success and responsiveness to defend against internal attacks in task orchestration \cite{trustattack2}, some studies have leveraged such scores and their resulting trust-based rankings for service provider selection and task offloading \cite{trusttaskoffloading1, trusttaskoffloading2, chen2023rte,joshuagreenjiazhichen}. For graph-structured tasks, \cite{chengoal2025} formulated the task orchestration problem as graph matching optimization to select an optimal set of trusted collaborators for task completion. However, the trust evaluation methods in this work is sensitive to weight settings of different trust factors and depend on prior knowledge, rendering it susceptible to inaccurate or biased evaluations. ML predict trust by capturing specific patterns in trust-related data \cite{11395598}, enabling automatic weight learning and achieving high accuracy with sufficient data. While Graph Neural Networks (GNNs) enable end-to-end trust evaluation from raw graph data \cite{trustgnn,Trustguard}, existing methods are typically designed for fixed graph structures and static trust, limiting their ability to model dynamically evolving, task-specific relationships and incurring high training overhead in large-scale systems, without supporting group-level, task-specific trust.

\textit{2) Learning-Based Task Orchestration:}
Motivated by the capability of DRL for automatic decision-making, extensive studies have been conducted to incorporate DRL in task orchestration \cite{DRLsurvey}, typically categorized into value-based and policy-based methods. For instance, Deep-Q-Network (DQN) learns the control policy via value function optimization \cite{11595216}, whereas deep deterministic policy gradient (DDPG) estimates policy gradients to directly optimize the control policy~\cite{DDPG}. However, these DRL methods rely on complete state information and neglect interactive dependencies between agents. To address these limitations, recent efforts have shifted toward state abstraction techniques, with GNN emerging as a powerful tool for modeling structural and temporal dependencies to boost training, such as adopting GNN to encode cluster topology or extract inter-subtask dependencies \cite{liwanggraph}. However, these GNN-based DRL assume fully observable and trustworthy environments, reducing adaptability to unseen task structures, thus incurring high training overhead. Some studies integrate trust into DRL \cite{trustDRL1, trustDRL2} by leveraging historical trust to guide decision-making. However, in these works, trust is treated as an auxiliary, task-agnostic variable within standard DRL focused on reward maximization, without capturing task-specific features and long-term fairness.

\textit{3) Fairness-Aware Task Orchestration:}
Building on existing fairness metrics, such as Jain’s index, entropy function, and min-max ratio, fairness-aware task orchestration has been widely explored in network optimization. Recent method in \cite{convexfairness2} overlooks long-term variations in task requirements and resource conditions and simplifies fairness optimization through convex assumptions, which is unrealistic in practical scenarios. To guarantee long-term fairness, \cite{dynamicfairness1} and \cite{dynamicfairness2} proposed dynamic resource allocation algorithms that rely on current system states without assuming prior distribution knowledge. However, these studies focus on specific fairness criteria (e.g., max-min or proportional fairness) rather than a general and tunable formulation such as $\alpha$-fairness, which encompasses diverse system fairness objectives. Furthermore, various DRL algorithms have been applied to fair task orchestration in domains such as cognitive radio networks \cite{fairheurstic2}, yet their heuristic design restricts their adaptability to $\alpha$-fairness in optimization. Some DRL algorithms have shown potential in optimizing $\alpha$-fairness \cite{fairmontecarlo}. However, they depend on Monte Carlo-based gradient estimation, resulting in slow convergence and remain decoupled from trust-aware and hierarchical decision-making.

\subsection{Contributions}
Motivated by these observations, we propose a fairness-aware concurrent task orchestration strategy that jointly performs trusted collaborator group formation and task allocation under dynamic and resource-constrained conditions. The contributions of this paper are summarized below.

\textbullet ~We design a dynamic task utility model, rather than relying solely on static performance metrics, to guide collaborator group formation and task allocation aligned with diverse task completion goals. Tasks are decomposed into limited subtasks under resource constraints while preserving inter-subtask dependencies, and collaborator groups are evaluated via a unified value-risk task utility that captures both goal satisfaction and deviation under dynamic conditions. Based on this task utility model, concurrent task orchestration is modeled as a long-term average utility maximization problem under $\alpha$-fairness, enabling joint optimization of performance and fairness beyond conventional formulations.

\textbullet ~We develop a hierarchical matching strategy that forms and matches trusted collaborator groups in a coarse-to-fine manner for scalable and adaptive task orchestration under long-term utility optimization. Instead of selecting individual collaborators, the strategy models group-level capability and reliability for task completion. To improve scalability in large-scale systems, tasks and collaborators are first clustered into coarse-grained sets based on feature similarity, within which fine-grained group formation and subtask allocation are performed, reducing mismatches between dynamic task specifics and varying collaborator conditions.

\textbullet ~To support the hierarchical matching strategy, we integrate task-specific trust prediction with fairness-aware policy learning into a unified framework. Specifically, GNNs infer task-specific trust by capturing spatial-temporal dependencies between tasks and collaborators across both coarse-grained and fine-grained levels. These trust representations are then embedded into a fairness-aware DRL framework to guide long-term matching policies. Each matching process is formulated as a Markov Decision Process (MDP), where predicted trust values serve as abstracted states and fairness-aware rewards are derived by adjusting task utilities based on fairness objectives and historical reward statistics.

The remainder of this paper is structured as follows. Section \ref{sectionii} describes the system model and problem formulation. Section \ref{sectioniii} details our proposed strategy. Simulation results are presented in Section \ref{sectioniv}, and Section \ref{sectionv} concludes the paper.

\section{System Model and Problem Formulation \label{sectionii}}
As shown in Fig.~\ref{fig:pictures1}, we consider a collaborative system with UAVs $\mathcal{M}=\{1,\dots,M\}$, UGVs $\mathcal{N}=\{1,\dots,N\}$, and ground eavesdroppers (Eves) $\mathcal{A}=\{1,\dots,A\}$ intercepting Air-to-Ground (A2G) communications~\cite{communicationchannel,eveattack1}. For UAV $m$, $\tilde{\mathcal N}$ denotes the UGVs within its communication range. Each UAV continuously generates tasks over $\mathcal{T}=\{1,\dots,T\}$ and offloads them to UGV groups, while nearby Roadside Units (RSUs) collect task and collaborator information for decision-making. Major notations are summarized in Table~\ref{tab1111}.

\begin{figure}[!ht]
\centering
\includegraphics[width=0.45\textwidth]{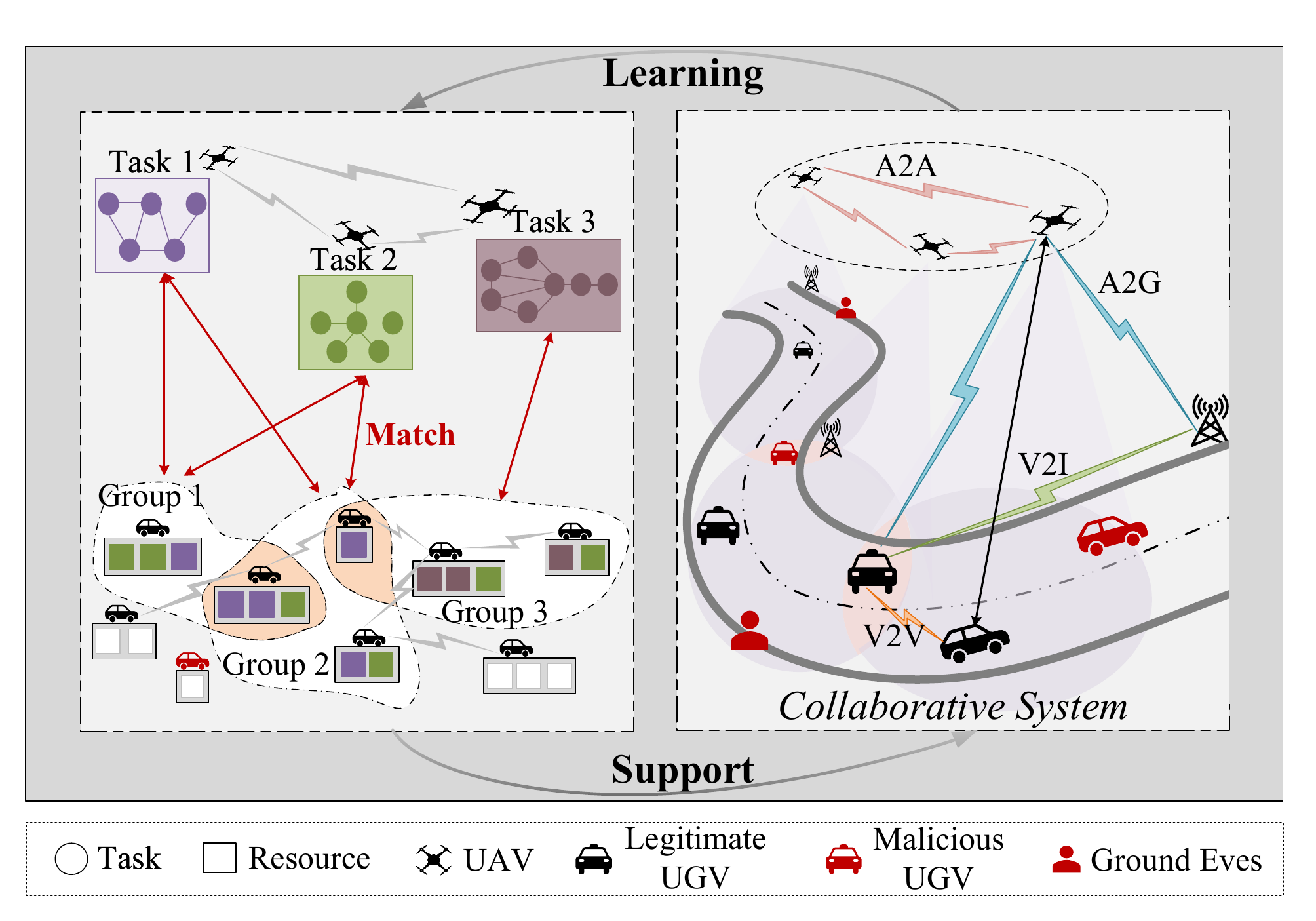}
\caption{An illustration of the UAV-UGV collaborative system, where UAVs offload tasks to trusted UGV groups.}
\label{fig:pictures1}
\end{figure}

\subsection{Modeling Tasks and Resources as Graphs}
Each task generated by UAV $m$ is decomposed into subtasks and modeled as an undirected graph
$\mathcal{G}_m=(\mathcal{S}_m,\mathcal{E}_m,\mathcal{W}_m)$
to preserve inter-subtask dependencies~\cite{liwanggraph,tasktype}, where
$\mathcal{S}_m=\{s_{im}|i\in\{1,\dots,|\mathcal{S}_m|\}\}$,
$\mathcal{E}_m=\{e_m^{ii'}|i,i'\in\{1,\dots,|\mathcal{S}_m|\},i\neq i'\}$,
and
$\mathcal{W}_m=\{w_m^{ii'}|i,i'\in\{1,\dots,|\mathcal{S}_m|\},i\neq i'\}$
denote subtasks, dependencies, and required contact durations, respectively.
For $M$ concurrent tasks,
$\mathcal{S}_{\mathcal M}=\cup_{m\in\mathcal M}\mathcal{S}_m
=\{s_l|l=(i,m)\}$ denotes the set of all subtasks.

Interdependent subtasks require reliable V2V communication among collaborating UGVs. We therefore model the UGV platform as a resource graph
$\mathcal{G}_{\mathcal N}=(\mathcal N,\mathcal E_{\mathcal N},\mathcal W_{\mathcal N})$,
where $\mathcal E_{\mathcal N}$ contains one-hop V2V links and
$\mathcal W_{\mathcal N}$ contains their exponential contact-duration rate parameters~\cite{liwanggraph}.
The resources of UGV $n$ are represented by VMs
$\mathcal V_n=\{v_{jn}|j\in\{1,\dots,|\mathcal V_n|\}\}$~\cite{chengoal2025},
with different CPU capabilities.
The complete resource set is
$\mathcal V_{\mathcal N}=\cup_{n\in\mathcal N}\mathcal V_n
=\{v_k|k=(j,n)\}$.
\begin{table}[!ht]
	\renewcommand{\arraystretch}{1}
        \caption{Major Notations and Definitions}
	\label{tab1111}
	\centering
	\begin{tabular}{ll}
		\hline\hline \textbf{Notation} & \textbf{Definition}\\
		\hline 
        $a$ & An Eve, $a \in \mathcal{A}$\\
        $e^{ii^{'}}_m$& A dependency (edge) between $i$ and $i^{'}$, $e^{ii^{'}}_m \in \mathcal{E}_m$\\
       $e_{nn'}$ & V2V link between $n$ and $n'$, $e_{nn'}\in\mathcal E_{\mathcal N}$\\
        $m$ & A UAV, $m \in \mathcal{M}$\\
      $n$ & A UGV, $n \in \mathcal{N}$\\
         $\hat{n}$& A collaborator within a collaborator cluster, $\hat{n} \in C_y$\\
         $\mathcal{O}_m$& A set of completion goals of a task $\mathcal{G}_m$\\
   $\boldsymbol{o}_{im}$& A set of completion goals of a subtask $s_{im}$\\
         $Q_m$& Utility of individual task $\mathcal{G}_m$\\
         $Q_x$& Utility of a cluster centroid\\
           $Q_p$& Averaged utility of all subtask within cluster $C_x$\\
           $R_m$& Risk of task completion\\
       $s_{im}$& A subtask of a UAV graph task, $s_{im} \in \mathcal S_m$\\
          $s_p$& A subtask within a task cluster, $s_{p} \in C_x$\\
                $U_m$& Value of task completion\\
                $v_k$& An available virtual machine, $v_k \in \mathcal{V}_{\mathcal{N}}$\\
	 $w^{ii^{'}}_m$& Required contact duration between $i$ and $i^{'}$, $w^{ii^{'}}_m \in \mathcal{W}_m$\\
  $w_{nn'}$ & Contact-duration rate parameter, $w_{nn'}\in\mathcal W_{\mathcal N}$\\
   $w_{x \rightarrow y}$& Trust value between $C_x$ and $C_y$\\
   $w_{s_p \rightarrow \hat{n}}$& Trust value between $s_p$ and $\hat{n}$\\
   $x_{im}^n$& Binary task offloading decision between $s_{im}$ and $n$\\
\hline\hline
	\end{tabular}
\end{table}

\subsection{Modeling Value of Task Completion}
As introduced in Section \ref{sectionImotivation}, the task completion value is the benefit derived from meeting diverse task completion goals. Accordingly, a tailored metric, Value of Task Completion~\cite{botaohyper,convexfairness2}, is utilized to evaluate how effectively the selected collaborator group fulfills diverse task completion goals in a unified manner. As the task requestor can directly perceive the time and energy costs of completing its tasks, value of task completion comprises both the value of time and the value of energy. Besides, a nearby RSU transfers its locally perceived reliability of selected collaborator groups to task requestor, which is termed the value of reliability. Thus, the value of task completion for subtask $s_{im}$ and overall task $\mathcal{G}_m$ is formulated as
\begin{equation}
    u_{im}^n=\varsigma_1 u_{im}^{\text{time}}+\varsigma_2 u_{im}^{\text{energy}}+\varsigma_3 u_{im}^{\text{reliability}},
\end{equation}
and
\begin{equation}
    U_{m}=\varsigma_4 U_{m}^{\text{time}}+\varsigma_5 U_{m}^{\text{energy}}+\varsigma_6 U_{m}^{\text{reliability}},
\end{equation}
respectively, where $\varsigma_{1-6}$ are weighting parameters for different task objectives. The task-level completion goals are $\mathcal O_m=[T_m^{\max},E_m^{\max},\Theta_m^{\max}]$, while each subtask has
$\boldsymbol{o}_{im}=[t_{im}^{\max},e_{im}^{\max},\theta_{im}^{\max}]$. The relationship between subtask goals and overall task goals is defined as follows: $T^{\text{max}}_m=\max\limits_{1\leq i\leq |\mathcal{S}_m|}t^{\text{max}}_{im}$, $E^{\text{max}}_m=\sum_{i=1}^{|\mathcal{S}_m|}e^{\text{max}}_{im}$, and $\Theta^{\text{max}}_m=\sum_{i=1}^{|\mathcal{S}_m|}\theta^{\text{max}}_{im}$, where the slowest subtask determines task completion time, while the overall energy consumption and reliability requirements are obtained by aggregating corresponding subtask demands.

\subsubsection{Value of Time}
The completion time $t_{im}^n$ includes transmission and computation delays, as well as V2V data-exchange delay when dependent subtasks are assigned to different collaborators:
\begin{equation} \label{eqsubtasktime}
t_{im}^n \!=\! 
\begin{cases} 
 \sum\limits_{n=1}^{\mathcal{\tilde N}}x_{im}^n\left(\frac{ D_{im}}{\phi_{mn}}+\sum\limits_{n^{'}\in \mathcal{\tilde N}}\Delta \tau_{nn^{'}} +\frac{D_{im}\rho_{im}}{f_{jn}} \right), \\ 
\quad\quad\quad\; \text{if } e^{ii^{'}}_m=1, x_{in} \times x_{i^{'}n^{'}}=1, n \neq n^{'}, \\ 
\sum\limits_{n=1}^{\mathcal{\tilde N}}x_{im}^n\left(\frac{ D_{ik}}{\phi_{mn}}+\frac{D_{im}\rho_{im}}{f_{jn}}\right), \quad \text{otherwise},
\end{cases}
\end{equation}
where $D_{im}$ and $\rho_{im}$ denote the data size and required CPU cycles per bit, and $f_{jn}$ is the CPU frequency of VM $v_{jn}$. As different values of $f_{jn}$ reflect diverse computing capabilities, the set $\boldsymbol{f}_n=\{f_{jn}|j\in \{1, 2,\ldots, |\mathcal{V}_n|\}\}$ denotes all available computing resources of UGV $n$. Besides, $\phi_{mn}=W{\log _2}( {1 + \frac{{p_{m}}{g_{mn}}}{{{\sigma ^2}} }} )$ indicates the data transmission rate between UAV $m$ and the UGV $n$, where $p_{m}$ is the transmit power, $W$ is the bandwidth, and $g_{mn} = g_1 (d^{\text{A2G}}_{mn})^{-\eta}$ represents the channel gain \cite{communicationchannel}. $g_1$ is the reference channel gain at 1 m, $d^{\text{A2G}}_{mn}$ is the A2G distance, and $\eta$ is the path loss exponent. Moreover, the data exchange time $\Delta \tau_{nn^{'}}$ follows an opportunistic V2V contact model which obeys exponential distribution \cite{liwanggraph}, where $\text{Pr} ( \Delta \tau_{nn^{'}}\geq \Delta t|w_{nn^{'}})= e^{-\Delta t\times w_{nn^{'}}}$ and $\Delta t=|\frac{D_{im}}{\phi_{mn}}-\frac{D_{i^{'}m})}{\phi_{mn^{'}}}|+w^{ii^{'}}_m$.

To assess whether the selected collaborator $n$ meets $t^{\text{max}}_{im}$, the value of subtask completion time is defined as
\begin{equation}
u_{im}^{\text{time}} =
\begin{cases} 
1-\frac{t_{im}^n }{t^{\text{max}}_{im}}, & \text{if } t_{im}^n  < t^{\text{max}}_{im}, \\ 
0, & \text{if } t_{im}^n \geq t^{\text{max}}_{im}.
\end{cases}
\end{equation}

The overall completion time $T_m$ of a task $\mathcal{G}_m$ is calculated based on the slowest subtask completion time $t_{im}^n$, denoted as $T_m=\max\limits_{1\leq i\leq |\mathcal{S}_m|}t_{im}^n$. Therefore, the value of overall task completion time is denoted as $U_{m}^{\text{time}}=\min\limits_{1\leq i\leq |\mathcal{S}_m|} u_{im}^{\text{time}}$.

\subsubsection{Value of Energy}
The UAV energy consumption for offloading $s_{im}$ includes transmission and channel-holding energy~\cite{zhao2025against}, which is given by
\begin{equation}
e_{im}^n=\sum_{n=1}^{|\mathcal{\tilde N}|} x_{im}^n \left(\frac{p_{m} D_{im}}{\phi_{mn}} + \ell_{im}\right),
\end{equation}
where $\ell_{im}$ denotes the energy consumed during the channel holding period after transmitting subtask $s_{im}$.

Similar to the value of time, the value of energy for completing the subtask $s_{im}$ is determined by comparing the expected energy consumption $e^{\text{max}}_{im}$ with the actual energy consumption $e_{im}^n$, expressed as
\begin{equation}
u_{im}^{\text{energy}} =
\begin{cases} 
1-\frac{e_{im}^n}{e^{\text{max}}_{im}}, & \text{if } e_{im}^n < e^{\text{max}}_{im}, \\ 
0, & \text{if } e_{im}^n \geq e^{\text{max}}_{im}.
\end{cases}
\end{equation}

The overall energy consumption for offloading a task $\mathcal{G}_m$ is calculated as $E_m = \sum_{i=1}^{|\mathcal{S}_m|} e_{im}^n$ and the corresponding value of energy is denoted as $U_{m}^{\text{energy}}=\sum_{i=1}^{|\mathcal{S}_m|} \frac{e^{\text{max}}_{im}}{E^{\text{max}}_m} u_{im}^{\text{energy}}$.

\subsubsection{Value of Reliability}
In open and dynamic collaborative systems, adversarial infiltration is inevitable, posing security threats to task offloading and completion. We consider two security threats: malicious UGVs launching blackhole or on-off attacks~\cite{attack1,attack2}, and ground Eves intercepting UAV-to-UGV transmissions~\cite{communicationchannel,eveattack1}. Accordingly, collaborator reliability combines identity authentication and transmission secrecy. Specifically, the identity of a selected collaborator $n$ is verified through lightweight continuous authentication locally executed at RSUs \cite{fang2023joint}, with the authentication accuracy $\psi_{n}$ defined as the weighted probability of correctly authenticating the collaborator’s identity, given by  
\begin{equation} \label{autheneq1} 
\psi_{n}= 1-(\epsilon_1 P_{FA} + \epsilon_2 P_{MD}),  
\end{equation}  
where \( P_{FA} = \Pr(\widetilde{\gamma}_{n\hat{t}} \neq \gamma_{n\hat{t}} \mid \psi_0) \) and \( P_{MD} = \Pr(\widetilde{\gamma}_{n\hat{t}} = \gamma_{n\hat{t}} \mid \psi_1) \) denote the false alarm rate and misdetection rate, respectively, and \( \psi_0 \) and \( \psi_1 \) indicate whether $n$ is legitimate or malicious, respectively. The authentication results at time $t \in \mathcal{T}$ are determined as follows:
\begin{equation} 
\phi_n(t)=\left\{\begin{array}{l} \psi_0, \;\text{if } \widetilde{\gamma}_{n\hat{t}}=\gamma_{n\hat{t}}, \!\!\!\quad \forall \hat{t} \in\{1,2, ..., t\} \\ \psi_1, \;\text{if } \widetilde{\gamma}_{n\hat{t}} \neq \gamma_{n\hat{t}}, \!\!\!\quad \exists \hat{t} \in\{1,2, ..., t\}\end{array},\right. 
\end{equation} 
where \( \gamma_{n\hat{t}} \) and \( \widetilde{\gamma}_{n\hat{t}} \) denote the expected and actual access states toward nearby RSUs, respectively, within the access time sequences $\bm{\gamma}_{j} = [\gamma_{n1},\gamma_{n2},\ldots,\gamma_{nT} ]^\top$ and $\widetilde{\bm{\gamma}}_{n} = [\widetilde{\gamma}_{n1},\widetilde{\gamma}_{n2},\ldots,\widetilde{\gamma}_{nT} ]^\top$. $(\cdot)^\top$ denotes vector transposition.

The possibility of maintaining confidential transmission when selecting $n$ for completing $s_{im}$ is quantified as
\begin{equation} \label{secrecyeq}
\eta_{im}^n = \left\{ {\begin{array}{*{20}{l}}
{\frac{{{\delta_{im} R_{mn}^{\text{sec}}} }}{{{D_{im}}}}}&{{\rm{if}}{\mkern 1mu} {\mkern 1mu} \;\;{\delta_{im} R_{mn}^{\text{sec}}}  < D_{im}}\\
{1}&{{\rm{if}}{\mkern 1mu} {\mkern 1mu} \;\;{\delta_{im} R_{mn}^{\text{sec}}}  \ge D_{im}}
\end{array}}, \right.
\end{equation}
where $\delta_{im}=t^{\text{max}}_{im}-\frac{D_{im}\rho_{im}}{f_{jn}}$ 
is the maximum transmission time and \begin{small}$R_{mn}^{\text{sec}} = \underset{\substack{||g_{ma}-\hat{g}_{ma}|| \leq \Delta g_1 \\ ||g_{na}-\hat{g}_{na}|| \leq \Delta g_2}}{\min} \left[ \phi_{mn} - \underset{a \in \mathcal{A}}{\max} \phi_{ma} \right]^+$\end{small} is the worst-case achievable secrecy rate~\cite{eveattack1}. Here, \begin{small}$\phi_{ma}=W{{\rm{lo}}{{\rm{g}}_2}( {1 + \frac{{{p_{m}}g_{ma}^{\text{max}}}}{p_n^{\text{jam}}g_{na}^{\text{min}}+ {\sigma ^2}}})}$\end{small}, where $g_{ma}^{\max}=\hat g_{ma}+\Delta g_1$ and
$g_{na}^{\min}=\hat g_{na}-\Delta g_2$ characterize the worst-case channels, and $p_n^{\rm jam}$ denotes the jamming power.

Based on (\ref{autheneq1}) and (\ref{secrecyeq}), the reliability of the selected UGV $n$ in completing subtask $s_{im}$ is denoted as
\begin{equation}
\theta_{im}^n = \sum\limits_{n=1}^{|\mathcal{\tilde N}|} x_{im}^n \left(\zeta_1 \psi_{n}+ \zeta_2\eta_{im}^n \right),
\end{equation}
where $\zeta_1$ and $\zeta_2$ indicate the weight parameters that reflect the relative importance of authentication and secrecy.

Due to varying sensitivities of subtasks to security threats, the reliability requirements for the selected collaborator differ. Therefore, the value of reliability of a subtask $s_{im}$ is designed to adapt to task-specific requirements, given by
\begin{equation}
u_{im}^{\text{reliability}} =
\begin{cases} 
\frac{\theta_{im}^n}{\theta^{\text{max}}_{im}}, & \text{if } \theta_{im}^n < \theta^{\text{max}}_{im}, \\ 
1, & \text{if } \theta_{im}^n \geq \theta^{\text{max}}_{im}.
\end{cases}
\end{equation}

The overall reliability of the selected UGV groups for completing \begin{small}$\mathcal{G}_m$ is calculated as $\Theta_m = \sum_{i=1}^{|\mathcal{S}_m|} \theta_{im}^n$\end{small} and its value of reliability is denoted as \begin{small}$U_{m}^{\text{reliability}}=\sum_{i=1}^{|\mathcal{S}_m|} \frac{\theta^{\text{max}}_{im}}{\Theta^{\text{max}}_m} u_{im}^{\text{reliability}}$\end{small}.
\vspace{-0.2cm}
\subsection{Modeling Risk of Task Completion}
Dynamic task requirements and collaborator conditions introduce the risk of violating latency, energy, and reliability goals. Following the Risk of Task Completion metric in~\cite{jiazhiICC}, inspired by Conditional Value at Risk (CVaR)~\cite{risk1}, we define a task-specific trust threshold concerning dynamic subtask completion goals: the task requestor $m$ announces each subtask trust threshold of $1-\alpha_{im}$, where $\alpha_{im} = \xi_1 t^{\text{max}}_{im}+\xi_2 e^{\text{max}}_{im}+\xi_3 (1-\theta^{\text{max}}_{im})\in (0,1)$. $0\leq \xi_1, \xi_2, \xi_3\leq 1$ are weighting coefficients capturing the relative emphasis on different completion goals, with $\xi_1+\xi_2+\xi_3=1$, and all completion goals are normalized to a common scale between 0 to 1. The risk of selecting collaborator $n$ for completing subtask $s_{im}$ under the trust threshold of $1-\alpha_{im}$ is evaluated as
\begin{equation}
r_{\alpha_{im}}^{n}=\sum_{n=1}^{|\mathcal{\tilde N}|} x_{im}^n\inf\limits_{\beta\in \mathbb{R}}\left\{\beta+\mathbb{E}[\mathcal{Z}-\beta]^{+}\right\},
\end{equation}
where $\beta$ denotes an auxiliary variable with the same range as $\mathcal{Z}$, $\mathbb{R}$ represents the set of all real numbers, and $\mathbb{E}$ is the expectation. Based on subtask completion goals $\boldsymbol{o}_{im}$ and actual completion performance $\widetilde{\boldsymbol{o}}_{im}^n=[t_{im}^n, e_{im}^n, \theta_{im}^n]$ of selecting collaborator $n$, when $t_{im}^n > t^{\text{max}}_{im}$, $e_{im}^n > e^{\text{max}}_{im}$, and $\theta^{\text{max}}_{im} > \theta_{im}^n$, $\mathcal{Z}=w_1(t_{im}^n-t^{\text{max}}_{im})+w_2(e_{im}^n-e^{\text{max}}_{im})+w_3(\theta^{\text{max}}_{im}-\theta_{im}^n)$, otherwise, $\mathcal{Z}=0$, where $w_1, w_2, w_3$ scale the deviation from task requirements without affecting the underlying risk formulation. Note that all vectors in $\widetilde{\boldsymbol{o}}_{im}^n$ and $\boldsymbol{o}_{im}$ are normalized to a common scale between 0 to 1.

The overall risk for completing $\mathcal{G}_m$, denoted as \( \mathcal{R}_m \), considers both the risk of assigning different collaborators to individual subtasks and the compounded risk introduced by interdependencies between subtasks, given by
\begin{equation}
 R_m=\sum_{i=1}^{|\mathcal{S}_m|}w_{im}r_{\alpha_{im}}^{n}+ \sum_{i=1}^{|\mathcal{S}_m|} \sum_{i' > i} e^{ii'}_m w^{ii'}_m \max(r_{\alpha_{im}}^n, r_{\alpha_{i'm}}^n),
\end{equation}
where $w_{im}$ indicates the relative importance of $s_{im}$ in $\mathcal{G}_m$, reflecting task structure while preserving risk modeling.

\subsection{Problem Formulation}
To jointly capture the value and risk of collaborator group selection, the utility of individual task $\mathcal{G}_m$ is needed, given by $Q_m=\frac{U_m}{\varrho+R_m}$, where $U_m$ and $R_m$ denote the obtained value and perceived risk for completing $\mathcal{G}_m$, respectively, and the constant $\varrho=1$ indicates that the impact of $U_m$ and $R_m$ is balanced equally, with $R_m$ reducing $Q_m$ proportionally to its magnitude in a scale-invariant manner.

However, task requirements and collaborator conditions vary over time, the same collaborator group may yield different utility for identical tasks, motivating long-term average utility for stable decision-making. Given the individual task utility $Q_m(t)=\frac{U_m(t)}{\varrho+R_m(t)}$ at time $t$, the long-term average task utility is obtained by
\begin{equation}
\overline{Q}_m=\lim_{T \to \infty} \frac{1}{T} \sum_{t=1}^{T} Q_m(t).
\end{equation}

To incorporate fairness in concurrent task orchestration, we adopt an \textit{$\alpha$-fair utility function} to adjust $\overline{Q}_m$ based on fairness preferences ($\alpha \geq 0$), with smaller $\alpha$ favoring efficiency and larger $\alpha$ enforcing fairness~\cite{fairnesscomputingoffloading}, given by 
\begin{equation} \label{utitlityfunction}
U(x) =
\begin{cases} 
\frac{x^{1-\alpha}}{1-\alpha} & \text{if} \ \alpha \neq 1, \\
\log(x) & \text{if} \ \alpha = 1.
\end{cases}
\end{equation}

Based on the above system models, we formulate the objective as maximizing the long-term average utility of concurrent tasks under fairness considerations through trusted collaborator selection $\mathcal{X}=[x_{im}^n(t)]_{1 \leq t \leq T, 1 \leq n \leq N, 1 \leq m \leq M, 1 \leq i \leq |\mathbf S_m|}$ in resource-constrained collaborative systems, given by
\begin{align} \label{eqproblem}
    &\;\;\;\;\;\;\;\;\;\;\;\;\;\;\;\;\;\;\;\;\;\;\;\;\;\;\bm{{\cal P}}:\max\limits_{\mathcal{X}} \sum_{m = 1}^{M}U(\overline{Q}_m) \\
    & \text{s.t.} \;\;\;\;\;\; \sum_{n=1}^{N} {x_{im}^n(t)} = 1, \quad \forall i, m, t \in \left|\mathcal{S}_m\right|, {\cal M},{\cal T}, \tag{16a} \\
    & \;\;\;\;\;\;\;\;\;\;\sum_{m = 1}^{M}\sum_{i=1}^{\left|\mathcal{S}_m\right|} {x_{im}^n(t)} \le \left|\mathcal{V}_n\right|, \quad \forall n,t \in {\cal N},{\cal T}, \tag{16b} \\
    &\begin{aligned}
        & \;\;\;\;\;\;\;\;\;\;\;e^{-\Delta t\times w_{nn^{'}}} \geq \mu, \quad \forall {x_{im}^n(t)} \times {x_{i^{'}m}^{n^{'}}(t)}=1, \\ 
        & \;\;\;\;\;\;\;\;\;\;\;\;\;\forall n, n^{'}\in \mathcal{\tilde N}, n \neq n^{'}, \forall e^{ii^{'}}_m, m, t \in \mathcal{E}_m, {\cal M},{\cal T}
    \end{aligned} \tag{16c} \\
    &\begin{aligned}
        & \;\;\;\;\;\;\;\;\;\;\;p_{m}(t)=0, \quad {x_{im}^n(t)}=0, \\ 
        & \;\;\;\;\;\;\;\;\;\;\;\;\;\forall i, m, t \in \mathcal{S}_m, {\cal M},{\cal T}, n \notin \mathcal{\tilde N}
    \end{aligned} \tag{16d}
\end{align}
where $x_{im}^n(t)=1$ denotes collaborator $n$ is selected to complete subtask $s_{im}$ generated at time $t$, otherwise $x_{im}^n(t)=0$. Constraint (16a) states that each subtask can be offloaded to only one UGV. Constraint (16b) limits the offloaded subtasks to exceed the available VMs of each UGV. Constraint (16c) ensures that when connected subtasks $s_{im}$ and $s_{i^{'}m}$ of task $\mathcal{G}_m$ are offloaded to different UGVs $n$ and $n^{'}$, the probability of data exchange time $\Delta \tau_{nn^{'}}$ exceeding $\Delta t$ must surpass a threshold $\mu$ while Constraint (16d) restricts UAVs to offload subtasks only to UGVs within their communication radius.

\textbf{Remark 1:} Solving the mixed integer programming problem in (\ref{eqproblem}) is extremely challenging due to 1) the difficulty of long-term average optimization, induced by inherent uncertainties from dynamic task requirements, varying collaborator conditions, and diverse attacks in collaborative systems; and 2) the need to maximize the total task utility of all concurrent tasks, while simultaneously ensuring long-term fairness in resource allocation and allowing short-term fairness variations due to system dynamics such as task urgency and limited resources. These challenges motivate the hierarchical strategy developed next to solve \(P\).

\section{Hierarchical Trusted Collaborator Group Formation and Task Allocation}\label{sectioniii}
To solve (\ref{eqproblem}), we propose a hierarchical strategy for trusted collaborator group formation and task allocation under long-term fairness. As shown in Fig. \ref{fig:pictures2}, the framework consists of two core modules: 1) \underline{H}ierarchical \underline{T}ask-\underline{S}pecific \underline{T}rust prediction via \underline{GNN}s (HTST-GNN) at coarse- and fine-grained levels, and 2) trust-embedded fairness-aware DRL for coarse-to-fine matching. Online inference is coordinated by nearby RSUs, while model training is performed offline and periodically.

\subsection{Hierarchical Task-Specific Trust Prediction via GNNs \label{subsectionivA}}
Dynamic task requirements and collaborator conditions directly affect the objective in (\ref{eqproblem}), while representing all such factors cannot efficiently capture their interdependencies. We therefore use task-specific trust to abstract the value and risk of task-resource interactions. The proposed HTST-GNN dynamically coarsens tasks and resources and performs spatial-temporal aggregation to predict trust at both coarse- and fine-grained levels.
\begin{figure*}[!ht]
\centering
\includegraphics[width=0.70\textwidth]{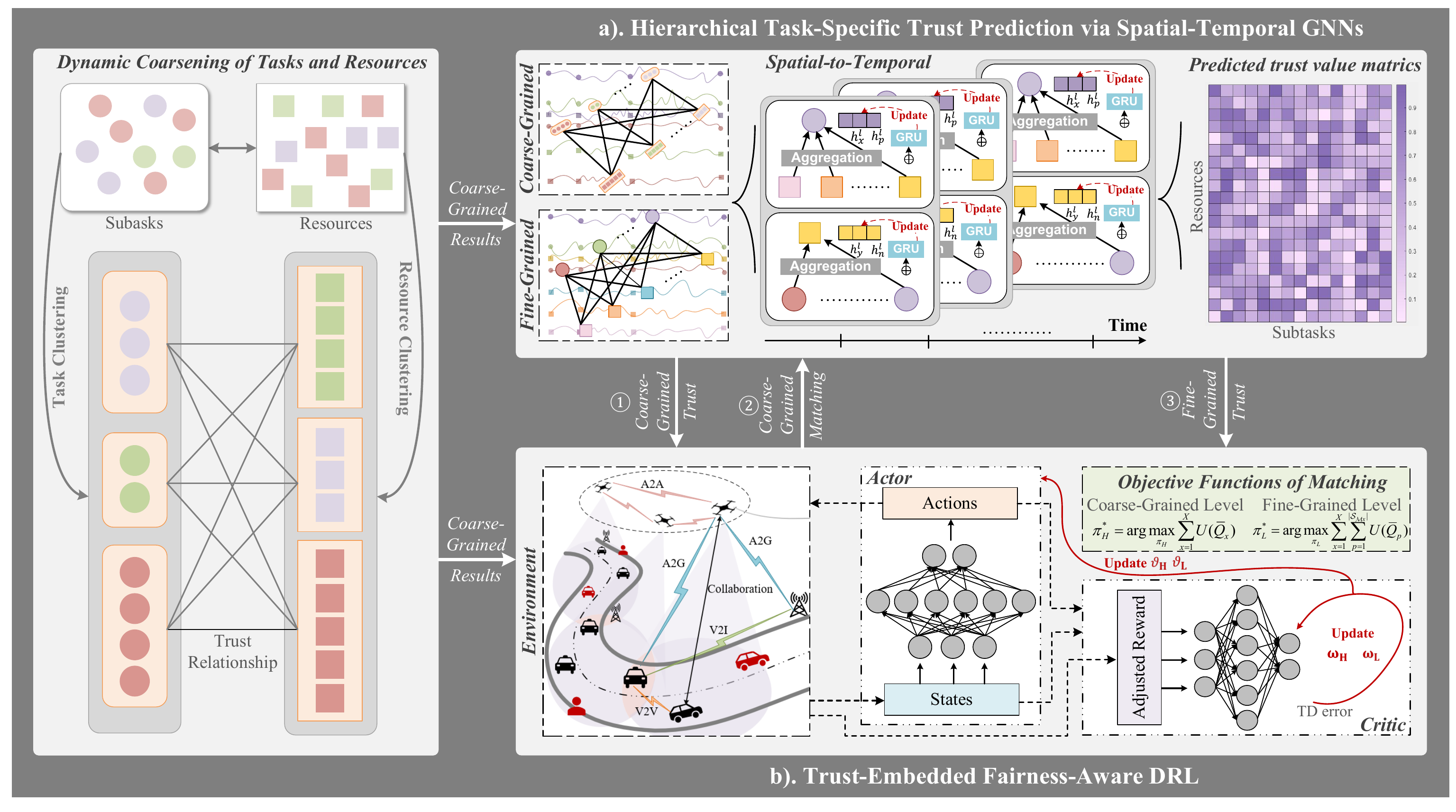}
\caption{Overview of the proposed hierarchical trusted collaborator group formation and task allocation strategy for concurrent task completion under long-term fairness. It consists of two modules: (i) Hierarchical Task-Specific Trust Prediction via spatial-temporal GNNs, and (ii) Trust-Embedded Fairness-Aware DRL.}
\label{fig:pictures2}
\end{figure*}

We model the collaborative system over time as a sequence of trust-guided bipartite graphs $G(t)=(\mathcal{S}_{\mathcal{M}}(t), \mathcal{N}(t), \mathcal{E}(t))$ during time $\{1,2,\ldots,t\}$, where $\mathcal{G}^{\hat{t}}=(\mathcal{S}_{\mathcal{M}}^{\hat{t}}, \mathcal{N}^{\hat{t}}, \mathcal{E}^{\hat{t}})$ belongs to $G(t)\rightarrow \{\mathcal{G}^1, \mathcal{G}^2,\ldots, \mathcal{G}^t\}, \forall \hat{t} \in\{1,2,\ldots, t\}$. Each $s_l^{\hat{t}} \in \mathcal{S}_{\mathcal{M}}^{\hat{t}}$, $n^{\hat{t}} \in \mathcal{N}^{\hat{t}}$, and $e_{s_l \rightarrow n}^{\hat{t}} \in \mathcal{E}^{\hat{t}}$ have time-varying features, i.e., subtask features $\mathcal{F}_l^{\hat{t}}=[\boldsymbol{o}_l(\hat{t}), D_l(\hat{t}), \rho_l(\hat{t}), \mathcal{G}_m(\hat{t})]$, $\forall l=(i,m), i\in \{1,2,\ldots,|\mathcal{S}_m|\}, m \in \mathcal{M}$, collaborator features $\mathcal{F}_n^{\hat{t}}= [\phi_n(\hat{t}), p_{n}^{\text{jam}}(\hat{t}), \boldsymbol{f}_n(\hat{t}), \mathcal{G}_\mathcal{N}(\hat{t})]$, and trust value $w_{s_l \rightarrow n}^{\hat{t}} \in \mathcal{W}^{\hat{t}}$. The objective is to predict trust values at \(t+1\) by training a trust evaluation model on \(G(t)\) in a coarse-to-fine manner, with coarse-grained prediction between task and resource sets followed by fine-grained prediction between subtasks and collaborators within matched sets.

\subsubsection{Dynamic Coarsening Process}
Dynamic coarsening partitions tasks into feature-similar sets to reduce task-collaborator mismatches and the subsequent matching space.

\begin{definition}[\textbf{Coarsening of Concurrent Tasks}]
\label{def2}
\textit{Given the graph representation of individual task $\mathcal{G}_m, \forall m \in \mathcal{M}$, the coarsening of concurrent tasks, referred to by $\mathcal{C} \left( \left\{ \mathcal{G}_m \right\}_{m \in \mathcal{M}} \right)$, is the clustering of all subtasks $\mathcal{S}_{\mathcal{M}}$ based on subtask feature similarities, given by}
\begin{small}
\[\mathcal{C} \left( \left\{ \mathcal{G}_m\right\}_{m \in \mathcal{M}} \right)\triangleq (\mathcal{C}_\mathcal{M}, \mathcal{E}_\mathcal{M}),\]
\end{small}\textit{where $\mathcal{C}_\mathcal{M}=\left\{C_x|x=\{1,2,\ldots, X\}\right\}$, $X$ denotes the number of clusters, $C_x$ represents the $x$-th cluster including one or more subtasks from different UAVs, i.e., $C_x = \left\{ s_{p} \mid p \in \{1, 2, \dots, |\mathcal{S}_{\mathcal{M}x}|\} \right\}$ with features of cluster centroids $\mathcal{F}_x=\frac{1}{|\mathcal{S}_{\mathcal{M}x}|}\sum_{s_p \in C_x} \mathcal{F}_p$, $s_{p}$ and $\mathcal{F}_p$ are the $p$-th subtask and its features in $C_x$, respectively, $\mathcal{S}_{\mathcal{M}x}=\cup_{m \in \mathcal{M}}\mathcal{S}_{mx}, \forall\mathcal{S}_{mx} \subseteq \mathcal{S}_m$ represents the subtasks assigned to $C_x$, and $\mathcal{E}_\mathcal{M}=\left\{\mathcal{E}_m\right\}_{m \in \mathcal{M}}$ denotes the set of subtask dependencies.}
\end{definition}

Similarly, heterogeneous collaborator identities, secrecy capabilities, and computing resources motivate coarsening the UGV set as follows.

\begin{definition}[\textbf{Coarsening of Distributed Resources}]
\label{def_coarsening_resources}
\textit{The coarsening of distributed resources, referred to by $\mathcal{C} \left( \mathcal{G}_\mathcal{N} \right)$, is a construction of different collaborator sets} via clustering based on collaborator feature similarities, given by
\begin{small}
\[\mathcal{C} \left( \mathcal{G}_\mathcal{N} \right) \triangleq (\mathcal{C}_\mathcal{N}, \mathcal{E}_\mathcal{N}),\]
\end{small}\textit{where $\mathcal{C}_{\mathcal{N}}=\{C_{y} \mid y =\{1,2,\dots, Y\}\}$, $Y$ denotes the number of clusters, and $C_{y}$ represents the $y$-th cluster including one or more collaborators, i.e., $C_{y}= \{ \hat{n} \mid \hat{n} \in \{1,2,\dots, |\mathcal{N}_{y}|\}\}$ with features of cluster centroids $\mathcal{F}_y=\frac{1}{|\mathcal{N}_{y}|}\sum_{\hat{n} \in C_y} \mathcal{F}_{\hat{n}}$, $\hat{n}$ and $\mathcal{F}_{\hat{n}}$ are the $\hat{n}$-th collaborator and its features in $C_y$, respectively, $\mathcal{N}_{y} \subseteq \mathcal{N} = \bigcup_{\hat{n} \in C_{y}} \mathcal{V}_{\hat{n}}$ denotes the set of UGVs belongs to the $y$-th cluster and $\mathcal{V}_{\hat{n}}$ indicates the set of VMs of $\hat{n}$.}
\end{definition}

Based on \textbf{\textit{Definitions 1 and 2}}, K-means clustering~\cite{chengoal2025} produces feature-similar task and resource sets, yielding the coarsened trust-guided bipartite graph $G_{\mathcal{C}}(t)=(\mathcal{C}_\mathcal{M}(t), \mathcal{C}_\mathcal{N}(t), \mathcal{E}_{\mathcal{C}_\mathcal{M}\rightarrow \mathcal{C}_\mathcal{N}}(t))$, where $\mathcal{G}_{\mathcal{C}}^{\hat{t}}=(\mathcal{C}_\mathcal{M}^{\hat{t}}, \mathcal{C}_\mathcal{N}^{\hat{t}}, \mathcal{E}_{\mathcal{C}_\mathcal{M}\rightarrow \mathcal{C}_\mathcal{N}}^{\hat{t}})$ is a coarsened graph at $\hat{t}$.

\subsubsection{Spatial-Temporal Feature Aggregation}
After obtaining the coarsened trust-guided bipartite graphs, spatial-temporal feature aggregation is applied to capture both static and dynamic information in graphs $G_{\mathcal{C}}(t)$. As nodes and edges in each $\mathcal{G}_{\mathcal{C}}^{\hat{t}}$ have different attribute vector dimensions, the initial embeddings of coarse-grained subtask sets, collaborator sets, and their edges are obtained by transforming their raw features of the cluster centroid into a common vector space as $ h_x^0= W_{x} \cdot \mathcal{F}_x$, $ h_{y}^0= W_{y} \cdot \mathcal{F}_y$, and $\varpi_{x \rightarrow y}^0 =W_{e} \cdot w_{x \rightarrow y}$, where $W_{x}$, $W_{y}$, and $W_{e}$ indicate the learnable parameter matrices, respectively, $w_{x \rightarrow y}=\frac{1}{|\mathcal{S}_{\mathcal{M}x}|\times|\mathcal{N}_{y}|}\sum_{s_p \in C_x}\sum_{\hat{n} \in C_y}w_{s_p \rightarrow \hat{n}}$ indicates the edge weight (trust value) between $C_x$ and $C_y$, and $w_{s_p \rightarrow \hat{n}}=1-(\xi_1 t^{\hat{n}}_{p}+\xi_2 e^{\hat{n}}_{p}+\xi_3 (1-\theta^{\hat{n}}_{p}))$ denotes the initial trust value between $s_p \in C_x$ and $\hat{n} \in C_y$, where $\xi_1,\xi_2,\xi_3$ balance different task attributes in the initial trust estimation.

To obtain spatial information for trust modeling in $\mathcal{G}_{\mathcal{C}}^{\hat{t}}$, we first learn latent representations of subtasks and collaborators by propagating and aggregating information from their neighbors \cite{trustgnn}. At each layer $l =\{1,2,\dots, L\}$, the process of the cluster centroid $s_x$ of a task set aggregates propagated data from collaborators is formulated as
\begin{equation}\label{subtaskagg}
    h_{x}^{l}=\text{UPD}\left(\text{AGG}(\{h_{y}^{l-1}||\varpi_{x \rightarrow y}^{l-1}, \forall y \in \mathcal{Y}\})\right),
\end{equation}
where $||$ is the concatenation operator, $h_{y}^{l-1}$ and $\varpi_{x \rightarrow y}^{l-1}$ denote the embeddings of the cluster centroid $n_y$ of a resource set and edge between $s_x$ and $n_y$ at layer $l-1$, respectively. Besides, \text{UPD}($\cdot$) and \text{AGG}($\cdot$) denote the update function and aggregation function, respectively. We use weighted-sum aggregation to incorporate edge information and GRUs as the update function.

Symmetrically, the resource-set centroid is updated as
\begin{equation}\label{collaboratoragg}
h_{y}^{l}=\text{UPD}\left(\text{AGG}(\{h_{x}^{l-1}||\varpi_{x \rightarrow y}^{l-1}, \forall x \in \mathcal{X}\})\right).
\end{equation}

Note that $h_{x}^{L}$ and $h_{y}^{L}$ represent the embeddings of $s_x$ and $n_y$ at the last layer $L$, respectively, and are expressed as $h_{x}^{\hat{t}}$ and $h_{y}^{\hat{t}}$ when static aggregation is performed at time $\hat{t}$.

To capture temporal variations, attention is applied to historical embeddings $\mathcal{H}_x^t=[h_x^1,h_x^2,...,h_x^t ]^\top$ and $\mathcal{H}_y^t=[h_y^1,h_y^2,...,h_y^t ]^\top$ yielding
\begin{equation}\label{subtaskfinal}
   h_{x}^{\text{final}}=\sum _{\hat{t}=1}^{t} a^{\hat{t},t}_{x} W_x^{V} h^{\hat{t}}_{x},
\end{equation}
\begin{equation}\label{collaboratorfinal}
   h_{y}^{\text{final}}=\sum _{\hat{t}=1}^{t} a^{\hat{t},t}_{y} W_y^{V} h^{\hat{t}}_{y},
\end{equation}
where $W_m^{V}$ and $W_j^{V}$ are learnable parameter matrices, $a^{\hat{t},t}_{x}$ and $a^{\hat{t},t}_{y}$ are normalized attention coefficients, which measure the relative importance of embeddings from each time slot $\hat{t}$ to the target time slot $t$ as follows
\begin{equation}\label{subtaskimpweight}
    a^{\hat{t},t}_{x} =\iota\left( \frac{\left(W_x^{Q} h^{t}_{x}\right) \left(W_x^{K} h^{\hat{t}}_{x}\right)^\top}{\sqrt{d}} \right),
\end{equation}
\begin{equation}\label{collaboratorimpweight}
    a^{\hat{t},t}_{y} =\iota\left( \frac{\left(W_y^{Q} h^{t}_{y}\right) \left(W_y^{K} h^{\hat{t}}_{y}\right)^\top}{\sqrt{d}} \right),
\end{equation}
where $\iota$ is the softmax function for normalization, $d$ is the dimension of each input embedding, $W_x^{Q}$, $W_x^{K}$, $W_y^{Q}$, and $W_y^{K}$ are learnable parameter matrices. 

\begin{algorithm}[!ht]
\footnotesize
\caption{HTST-GNN for State Representation}
\label{GNNalgorithm}
\KwIn{Trust-guided bipartite graphs $\mathcal{G}(t)$}
\KwOut{Predicted trust matrices 
$[\hat{\mathcal{W}}]_{X\times Y}^{t+1}$ and
$[\hat{\mathcal{W}}]_{|\mathcal{S_M}|\times \mathcal{N}}^{t+1}$}

\tcp{Dynamic Coarsening}
\For{$\mathcal{G}^{\hat t}\in\mathcal{G}(t)$}{
Coarsen tasks and resources to obtain $\mathcal{G}_{\mathcal C}^{\hat t}$\;
}
Store $\mathcal{G}_{\mathcal C}(t)$\;

\tcp{Coarse-Grained Trust Prediction}
Initialize $h_x^0,h_y^0,\varpi_{x\rightarrow y}^0$\;
\For{$\mathcal{G}_{\mathcal C}^{\hat t}\in\mathcal{G}_{\mathcal C}(t)$}{
Prune $\mathcal{G}_{\mathcal C}^{\hat t}$ via (\ref{similarity})\;
Obtain $h_x^{\hat t},h_y^{\hat t}$ via
(\ref{subtaskagg})--(\ref{collaboratoragg})\;
}
Predict $\hat w_{x\rightarrow y}^{t+1}$ via
(\ref{subtaskfinal})--(\ref{predicttrust})\;

\tcp{Fine-Grained Trust Prediction}
\For{$(C_x,\mathbf{M}(C_x))$}{
Apply the same spatial-temporal prediction to
$s_p\in C_x$ and $n\in\mathbf{M}(C_x)$\;
Obtain $\hat w_{p\rightarrow n}^{t+1}$\;
}
Store the coarse- and fine-grained trust matrices\;
\end{algorithm}

\subsubsection{Prediction and Training}
The final embeddings are mapped by an Multi-Layer Perception (MLP) to obtain the edge embedding, given by
\begin{equation}
\widehat{\varpi}^{t+1}_{x \rightarrow y} = \text{MLP} ( h_{x}^{\text{final}} || h_{y}^{\text{final}} ).
\end{equation}

The corresponding trust value is subsequently derived by transforming $\widehat{\varpi}^{t+1}_{x \rightarrow y}$ into scalar representation, formulated as
\begin{equation}\label{predicttrust}
\hat{w}^{t+1}_{x \rightarrow y} = \varkappa ( W_p \cdot \widehat{\varpi}^{t+1}_{x \rightarrow y} + b_p),
\end{equation}
where $\varkappa$ is a sigmoid function that maps the trust values to the range [0, 1], $W_p$ and $b_p$ are learnable parameter matrices.

To make the trust prediction process task-specific, training is restricted to graph knowledge from similar tasks, thus reducing excessive complexity influenced by the size and interdependent diversity of nodes and edges in trust-guided bipartite graphs. Thus, a sub bipartite graph of $\mathcal{G}_{\mathcal{C}}^{\hat{t}}$ relevant to the new subtask $s_l^{t+1}$ is obtained by pruning the coarse-grained task sets and edges into two subsets, denoted as \begin{small}$\widetilde{\mathcal{C}}_\mathcal{M}^{\hat{t}}=\{C_x^{\hat{t}}|\forall x \in \widetilde{X}, \Lambda_{lx}^{\hat{t},t+1}\geq \nu\}$\end{small} and \begin{small}$\widetilde{\mathcal{E}}_{\mathcal{C}_\mathcal{M}\rightarrow \mathcal{C}_\mathcal{N}}^{\hat{t}}=\{e_{x \rightarrow y}^{\hat{t}} |\forall x \in \widetilde{X}, \forall y \in Y, \Lambda_{lx}^{\hat{t},t+1}\geq \nu\}$\end{small}, where $\widetilde{X} \subseteq X$ indicates the pruned subset of coarse-grained task sets, $\nu$ represents a similarity threshold, and $\Lambda_{lx}^{\hat{t},t+1}$ is the similarity coefficient between new subtask $s_l^{t+1}$ and previous cluster centroid $s_x^{\hat{t}}$, which is calculated based on Cosine similarity, given by
\begin{equation} \label{similarity} 
\Lambda_{lx}^{\hat{t},t+1}=\frac{\left(h^{0}_{l}(t+1)\right)^\top \cdot h^{0}_{x}(\hat{t})}{\Vert h^{0}_{l}(t+1) \Vert _{2} \cdot \Vert h^{0}_{x}(\hat{t}) \Vert _{2}},
\end{equation}
where $\Vert \cdot \Vert _{2}$ indicates the $L_2$ norm of a vector.

The GNN parameters $\Psi$ are optimized using the cross-entropy loss, which is denoted as
\begin{equation}\label{loss}
\Psi^{*}=\arg\min_{\Psi}\mathcal{L}(\hat{w}^{t+1}_{x \rightarrow y},w^{t+1}_{x \rightarrow y}),
\end{equation}
where $w^{t+1}_{x \rightarrow y}$ indicates the ground truth, and the Adam optimizer is utilized. For highly similar tasks, learned parameters are reused directly, otherwise fine-tuning is performed.

\textbf{Remark 2:} After coarse-grained trust prediction, the matched resource sets $\mathbf{M}(C_x)$ form candidate collaborator pools for each task set $C_x$, The same spatial-temporal aggregation and task-specific training are then applied to subtasks $s_p \in C_x$ and collaborators $\hat{n} \in \mathbf{M}(C_x)$, yielding fine-grained trust predictions with value $\hat{w}^{t+1}_{p \rightarrow n}$. The procedure of HTST-GNN is summarized in \textbf{Algorithm \ref{GNNalgorithm}}.

\subsection{Trust-Embedded Fairness-Aware DRL}\label{subsectionivB}
We solve $\bm{{\cal P}}$ in (\ref{eqproblem}) through hierarchical DRL~\cite{DDPG, mao2024dl}, progressively reducing the task-collaborator matching space. A coarse-grained policy first matches task and resource sets, after which a fine-grained policy performs subtask-collaborator matching within the selected sets.

\textit{Stage I (Optimization of the Coarse-Grained Task Set-Resource Set Matching Policy):} The coarse-grained matching policy $\pi_H^*$ is learned after dynamic coarsening of tasks and resources, formulated as
\begin{align}\label{eqproblemsub1}
    &\;\;\;\;\;\;\;\;\;\;\;\;\;\;\;\;\;\bm{{\cal P}_1}: \pi_H^*=\arg\max\limits_{\pi_H} \sum_{x = 1}^{X}U(\overline{Q}_x)\\
    & \text{s.t., } \;\;\;\;\;\;\;\;\;\; \sum_{y=1}^{Y} {\mathbf{x}_{x}^y(t)} \geq 1,\forall C_x \in {\mathcal{C}_\mathcal{M}}, t \in {\cal T}, \tag{27a}\\
    & \;\;\;\;\;\;\;\;\;\;\;\;\;\;\;\; \sum_{x=1}^{X} {\mathbf{x}_{x}^y(t)} \geq 1,\forall C_y \in {\mathcal{C}_\mathcal{N}}, t \in {\cal T}, \tag{27b}\\
     & \;\;\;\;\;\;\;\;\left|\mathcal{S}_{\mathcal{M}x}\right| \leq \sum_{C_y \in \mathbf{M}(C_x)}\sum_{\hat{n}\in C_y}\left|\mathcal{V}_{\hat{n}}\right|,\forall C_x \in {\mathcal{C}_\mathcal{M}}, \tag{27c}
\end{align}
where $\overline{Q}_x=\lim_{T \to \infty} \frac{1}{T} \sum_{t=1}^{T} Q_x(t)$ denotes the long-term average utility of task set $C_x$, with $Q_x(t)=\frac{1}{|\mathcal{S}_{\mathcal{M}x}|}\sum_{s_p \in C_x} \frac{u_p(t)}{\varrho+r_{\alpha_{p}}(t)}$ aggregating subtask-level value $u_p(t)$ and risk $r_{\alpha_{p}}(t)$. Here, $\mathbf{x}_{x}^y(t)$ indicates the matching decision between $C_x$ and $C_y$. Constraints (27a)-(27c) permit many-to-many set matching while ensuring that the subtasks in each task set do not exceed the available VMs of its matched resource sets $\mathbf{M}(C_x)$.

\textit{Stage II (Optimization of the Fine-Grained Subtask-Collaborator Matching Policy Within Matched Task-Resource Sets):} Conditioned on the coarse-grained matching policy $\pi_H^*$, the fine-grained policy $\pi_L^*$ optimizes subtask-collaborator matching within selected sets, formulated as
\begin{align}\label{eqproblemsub2}
    &\;\;\;\bm{{\cal P}_2}: \pi_L^*=\arg\max\limits_{\pi_L} \sum_{x = 1}^{X}\sum_{p = 1}^{|\mathcal{S}_{\mathcal{M}x}|}U(\overline{Q}_p)\\
    & \;\;\;\;\;\;\;\;\;\;\;\;\;\;\;\;\;\;\;\;\;\text{s.t., } (\text{16a})-(\text{16d}),\notag
\end{align}
where $\overline{Q}_p=\lim_{T \to \infty} \frac{1}{T} \sum_{t=1}^{T} \frac{u_p(t)}{\varrho+r_{\alpha_{p}}(t)}$ denotes the long-term utility of subtask $s_p$.

\subsubsection{Stage I–Design and Learning of Coarse-Grained Level Matching Policy} \label{subsectionBstageI}
The collaborative system and task sets serve as environment and learning agents, respectively, optimizing $\pi_H^*$ to maximize the task set rewards in an infinite-horizon discounted MDP. The three key components of MDP in $\bm{{\cal P}_1}$ are elaborated as follows:

a). \textit{State}: The coarse-grained state is the predicted task-specific trust matrix, defined as $\mathbf{s}_t^H=[s_1^H(t), \ldots, s_X^H(t)]$, where $s_x^H(t)=[\hat{w}^{t}_{x \rightarrow 1}, \ldots, \hat{w}^{t}_{x \rightarrow Y}]$ denotes the predicted trust values of task set $C_x$ over all resource sets at time $t$.

b). \textit{Action}: The action $\mathbf{a}_t^H=[a_1^H(t),...,a_X^H(t)]$ represents task-set/resource-set matching decisions, where $a_x^H(t)=[\mathbf{x}_{x}^1(t),...,\mathbf{x}_{x}^Y(t)]$.

c). \textit{Rewards}: Because long-term utility is unavailable online, the instantaneous reward $r_x^t=Q_x(t)$ is weighted by the derivative of the $\alpha$-fair utility in (\ref{utitlityfunction}) using the running average $\bar{h}_x^t=\frac{1}{t}\sum_{\hat{t}=1}^t r_x^{\hat{t}}$. The adjusted reward is then given by
    \begin{equation}\label{adjustrewardgroup}
\widetilde{r}_x^t=r_x^t{U^\prime }(\bar{h}^t_x),
\end{equation}
where $U^\prime(\cdot)$ denotes the derivative of the fairness utility, assigning higher weights to tasks with lower long-term utility.

At time $t$, task sets observe the state $\textbf{s}_t^H$ and select an action $\textbf{a}_t^H$ following the policy $\pi(a|s)$, which is represented by Artificial Neural Networks (ANNs) ${\pi}_{\boldsymbol{\vartheta_H}}$, where $\boldsymbol{\vartheta_H}$ indicates the policy parameter associated with ANNs. The following theorem establishes the stationary-point property of the resulting policy.

\begin{theorem}[\textbf{Policy Parameter $\boldsymbol{\vartheta_H}^*$ is the Stationary Point of the $\alpha$-Fair Utility Optimization}]
By following the optimal policy ${\pi}_{\boldsymbol{\vartheta_H}^*}$, the gradient (or derivative) of the $\alpha$-fair utility optimization problem equals zero, i.e., $\nabla _{\boldsymbol{\vartheta_H}} [{\sum_{x = 1}^{X}U(\overline{Q}_x)}]|_{\boldsymbol{\vartheta_H} =\boldsymbol{\vartheta_H} ^{\ast}} = 0$.
\end{theorem}
\begin{IEEEproof}
   Detailed derivations are given in \textbf{Appendix A}.
\end{IEEEproof}

Based on \textbf{\textit{Theorem 1}}, we propose a trust-embedded fairness-aware DRL algorithm with actor-critic structure to solve the $\alpha$-fair utility optimization in (\ref{eqproblemsub1}), in which the actor neural network selects actions according to the policy, while the critic neural network evaluates actions and guides policy updates.

a). \textbf{Actor Neural Network}: Employing the policy gradient method, the policy gradient of the objective $\sum_{x = 1}^{X}U(\overline{Q}_x)$ for $\bm{{\cal P}_1}$ is given by
\begin{align}
    & \nabla _{\boldsymbol{\vartheta_H}} \left[{\sum_{x = 1}^{X}U(\overline{Q}_x)}\right] = \notag\\
    &\;\;\;\;\;\;\;\;\;\;\;\;\;\mathbb{E}{_{{\pi}_{\boldsymbol{\vartheta_H}}} }\left[ \hat{A}_{\pi_{\boldsymbol{\vartheta_H}}}(s, a) \nabla_{\boldsymbol{\vartheta_H}} \ln \pi_{\boldsymbol{\vartheta_H}} (a|s)\right],
    \label{gradientgroup}
\end{align}
where $\hat{A}_{\pi_{\boldsymbol{\vartheta_H}}}(s, a)=\hat{\mathbb{Q}}_{\pi_{\boldsymbol{\vartheta_H}}}(s, a)-\hat{\mathbb{V}}_{\pi_{\boldsymbol{\vartheta_H}}}(s)$ denotes the advantage function that quantifies an action deviation from the expected state value. Here, \begin{small}
\begin{equation}\hat{\mathbb{Q}}_{\pi_{\boldsymbol{\vartheta_H}}}(s, a)= \mathop {\lim }\limits_{T \to \infty } \mathbb{E}{_{{\pi}_{\boldsymbol{\vartheta_H}}} }\left[ {\sum\limits_{t = 1}^T {\sum\limits_{x = 1}^X {\widetilde{r}_x^t} } - {{\left. {\hat{\varrho}_{\pi_{\boldsymbol{\vartheta_H}}}} \right|}_{{s_0} = s,{a_0} = a}}} \right]\notag\end{equation}
\end{small}and \begin{small}
\begin{equation} \hat{\mathbb{V}}_{\pi_{\boldsymbol{\vartheta_H}}}(s)= \mathop {\lim }\limits_{T \to \infty } \mathbb{E}{_{{\pi}_{\boldsymbol{\vartheta_H}}} }\left[ {\sum\limits_{t = 1}^T {\sum\limits_{x = 1}^X {\widetilde{r}_x^t} } - {{\left. {\hat{\varrho}_{\pi_{\boldsymbol{\vartheta_H}}}} \right|}_{{s_0} = s}}} \right] \notag\end{equation}
\end{small}represent the relative state-value and action-value functions, both centered by the average reward $\hat{\varrho}_{\pi_{\boldsymbol{\vartheta_H}}}=\lim_{T \to \infty} \frac{1}{T} \sum_{t=1}^{T}\sum_{x=1}^{X}\widetilde{r}_x^t$. Then, the policy gradient update for $\boldsymbol{\vartheta_H}$ is given by
    \begin{equation}
        \boldsymbol{\vartheta_H}  \leftarrow \boldsymbol{\vartheta_H} + \eta_{\boldsymbol{\vartheta_H}}  \cdot \nabla _{\boldsymbol{\vartheta_H}} \left[{\sum_{x = 1}^{X}U(\overline{Q}_x)}\right],
        \label{actorupdategroup}
    \end{equation}
    where $\eta_{\boldsymbol{\vartheta_H}} > 0$ is the learning rate for the policy update.
    
b). \textbf{Critic Neural Network}: The critic estimates the relative state-value function $\hat{\mathbb{V}}_{\pi_{\boldsymbol{\vartheta_H}}}(s)$ to evaluate the policy ${\pi}_{\boldsymbol{\vartheta_H}}$ and to guide actor updates. The critic approximates $\hat{\mathbb{V}}_{\pi_{\boldsymbol{\vartheta_H}}}(s)$ with $\hat{\mathbb{V}}_{\boldsymbol{\omega_H}}(s)$ and are optimized by minimizing the Mean Squared Error (MSE) loss, formulated as
\begin{equation}
    L\left( \boldsymbol{\omega_H} \right) = \mathbb{E}\left\{ {{{\left[ \sum_{x=1}^X \widetilde{r}_x^t + \gamma_0\hat{\mathbb{V}}_{\boldsymbol{\omega_H}}(\textbf{s}_{t+1}) - \hat{\mathbb{V}}_{\boldsymbol{\omega_H}}(\textbf{s}_t) \right]}^2}} \right\},
    \label{groupcriticloss}
\end{equation}
where $\gamma_0$ is the discount factor. The parameters $\boldsymbol{\omega_H}$ are iteratively updated as \begin{equation} \label{criticgroupupdate}
\boldsymbol{\omega_H} \leftarrow {\boldsymbol{\omega_H}}{\rm{ + }}\eta_{\boldsymbol{\omega_H}}  \cdot {{\nabla}_{\boldsymbol{\omega_H}}}L\left( \boldsymbol{\omega_H} \right),
\end{equation}where $\eta_{\boldsymbol{\omega_H}} > 0$ is the learning rate for $\boldsymbol{\omega_H}$.
\begin{algorithm}[!ht]
\fontsize{7.3pt}{6pt}\selectfont
\caption{Trust-Embedded Fairness-Aware DRL}
\label{DRLalgorithm}
\KwIn{Initial states $\mathbf{s}_t^H,\mathbf{s}_t^L$, learning rates, $\gamma_0$}
\KwOut{Coarse-grained policy ${\pi}_{\boldsymbol{\vartheta_H}^*}$ and fine-grained policy ${\pi}_{\boldsymbol{\vartheta_L}^*}$}

Initialize actor and critic parameters $\boldsymbol{\vartheta_H}, \boldsymbol{\omega_H}, \boldsymbol{\vartheta_L}, \boldsymbol{\omega_L}$

\For{$I=1:T/\Delta T$}{
    \For{$t=1:\Delta T$}{
        Sample $\mathbf{a}_t^H$; obtain $\mathbf{s}_{t+1}^H$;
        compute $\widetilde r_x^t$ via Eq.~(\ref{adjustrewardgroup});
        update $\boldsymbol{\vartheta_H},\boldsymbol{\omega_H}$
        via Eqs.~(\ref{actorupdategroup})--(\ref{criticgroupupdate})\;
    }

    Extract $(C_x,\mathbf{M}(C_x))$, $\forall x\in X$\;

    \For{$t=1:\Delta T$}{
        Sample $\mathbf{a}_t^L$; obtain $\mathbf{s}_{t+1}^L$;
        compute $\widetilde r_{p,x}^t$ via
        Eq.~(\ref{adjustrewardindividual});
        update $\boldsymbol{\vartheta_L},\boldsymbol{\omega_L}$
        following the Stage-I actor--critic updates\;
    }
}
\end{algorithm}

The optimal matching policy between task sets and resource sets ${\pi}_{\boldsymbol{\vartheta_H}^*}$ is obtained upon the objective function $\bm{{\cal P}_1}$ convergence after multiple iterations $I$, where each iteration runs for $ \Delta T$ time steps and the learning procedure is summarized in lines 3-5 of \textbf{Algorithm \ref{DRLalgorithm}}.

\subsubsection{Stage II–Design and Learning of Fine-Grained Level Matching Policy Within Matched Task-Resource Sets}
Under the converged coarse-grained level matching policy ${\pi}_{\boldsymbol{\vartheta_H}^*}$, which determines the candidate resource sets for each task set, the optimal matching results $\textbf{a}_t^{H^*}$ between task sets and resource sets are obtained. Based on these matched sets, the fine-grained policy ${\pi}_{\boldsymbol{\vartheta_L}^*}$ further optimizes the matching between individual subtasks and collaborators to solve $\bm{{\cal P}_2}$ under long-term fairness.

a). \textit{State}: Based on the trust prediction results between subtasks and collaborators within matched sets, the state is defined as\begin{small}
\begin{equation}\mathbf{s}_t^L = \big[ 
\overbrace{( s^L_{1, x}(t), \ldots, s^L_{|\mathcal{S}_{\mathcal{M}x}|, x}(t) )}^{C_x}, \forall C_x \in {\mathcal{C}_\mathcal{M}} \big],\notag
\end{equation}
\end{small}where the subvector is given by $s^L_{p, x}(t)=[\hat{w}^{t}_{p \rightarrow 1},..., \hat{w}^{t}_{p \rightarrow |\mathbf{M}(C_x)|}]$ and $|\mathbf{M}(C_x)|$ denotes the number of collaborators in $\mathbf{M}(C_x)$. 

b). \textit{Action}:  
Each subtask in $C_x$ is assigned to one collaborator from $\mathbf{M}(C_x)$, where the action space is restricted to matched sets, reducing the search space from $\mathcal{O}(N^{|\mathbf{S}_{\mathcal{M}}|})$ to $\mathcal{O}(N^{|\mathbf{M}(C_x)|})$, $|\mathbf{M}(C_x)|\ll |\mathbf{S}_{\mathcal{M}}|$. The fine-grained level action set is given by\begin{small}
\begin{equation}\mathbf{a}_t^L = \big[ 
\overbrace{( a^L_{1, x}(t), \ldots, a^L_{|\mathcal{S}_{\mathcal{M}x}|, x}(t) )}^{C_x}, \forall C_x \in {\mathcal{C}_\mathcal{M}} \big],\notag
\end{equation}
\end{small}where the matching action for $s_p \in C_x$ follows $a^L_{p, x}(t)=[x^1_{p,x}(t),..., x^{|\mathbf{M}(C_x)|}_{p,x}(t)]$.

c). \textit{Rewards}: Similar to Stage I, the original reward $r_{p,x}^t=\frac{u_p(t)}{\varrho+r_{\alpha_{p}}(t)}$ is adjusted under long-term fairness, yielding
\begin{equation}\label{adjustrewardindividual}
\widetilde{r}_{p,x}^t=r_{p,x}^t{U^\prime }(\bar{h}^t_{p,x}),
\end{equation}
where $\bar{h}^t_{p,x}=\frac{1}{t}\sum_{\hat{t}=1}^t r_{p,x}^{\hat{t}}$, $\forall \hat{t} \in\{1,2, ..., t\}$ represents the average historical reward for subtask $s_p$.

Based on the above MDP, the fine-grained policy ${\pi}_{\boldsymbol{\vartheta_L}}$ follows the same actor-critic learning procedure as Stage I, with $(\mathbf{s}_t^H, \mathbf{a}_t^H, \widetilde{r}_x^t, \boldsymbol{\vartheta_H}, \boldsymbol{\omega_H})$ replaced by fine-grained counterparts $(\mathbf{s}_t^L, \mathbf{a}_t^L, \widetilde{r}_{p,x}^t, \boldsymbol{\vartheta_L}, \boldsymbol{\omega_L})$. Actor and critic updates therefore follow (\ref{gradientgroup})–(\ref{criticgroupupdate}) within the matched task-resource sets. The overall learning procedure of the fine-grained level matching policy ${\pi}_{\boldsymbol{\vartheta_L}}$ is stated in lines 6-7 of \textbf{Algorithm \ref{DRLalgorithm}}. 

\subsection{Computational Complexity Analysis}\label{subsectionivC}

The dynamic coarsening process incurs
$\mathcal{O}(|\mathcal{S}_{\mathcal{M}}|X|\mathcal{F}|I
+|\mathcal{N}|Y|\mathcal{F}|I)$,
where $|\mathcal{F}|$ and $I$ denote the feature dimension and
the number of K-means iterations, respectively. Since
$X\ll|\mathcal{S}_{\mathcal{M}}|$ and $Y\ll|\mathcal{N}|$,
this introduces only a lower-order cost.

The coarse-grained complexity is
\begin{small}
\begin{equation*}
\mathcal{O}\left(
\begin{aligned}
&\underbrace{
L|{\cal E}_{{\cal C}_{\cal M}\to{\cal C}_{\cal N}}|d
+t(X+Y)d^2
}_{\text{Trust Prediction}}\\
&+\underbrace{
T_{\rm coarse}
\left(
\sum_{i=0}^{L_{\rm pol}}
n_{\rm pol}^{(i)}n_{\rm pol}^{(i+1)}
+
\sum_{j=0}^{L_{\rm val}}
n_{\rm val}^{(j)}n_{\rm val}^{(j+1)}
\right)
}_{\text{Matching}}
\end{aligned}
\right),
\end{equation*}
\end{small}
where the two terms correspond to spatial-temporal trust
prediction~\cite{trustgnn,Trustguard} and actor-critic
matching~\cite{DDPG}, respectively.

At the fine-grained level, computation is restricted to the
matched task-resource sets, yielding
\begin{small}
\begin{equation*}
\mathcal{O}\left(
\begin{aligned}
&\underbrace{
L\sum\nolimits_x^X
|{\cal E}_{p\in C_x\to n\in C_y}|d
+t(|{\cal S}_{\cal M}|+|{\cal N}|)d^2
}_{\text{Trust Prediction}}\\
&+\underbrace{
T_{\rm fine}
\left(
\sum_{i=0}^{L_{\rm pol}}
n_{\rm pol}^{(i)}n_{\rm pol}^{(i+1)}
+
\sum_{j=0}^{L_{\rm val}}
n_{\rm val}^{(j)}n_{\rm val}^{(j+1)}
\right)
}_{\text{Matching}}
\end{aligned}
\right),
\end{equation*}
\end{small}where $\sum_x^X|{\cal E}_{p\in C_x\to n\in C_y}|$ denotes the
number of edges within the matched sets and
$T_{\rm coarse}+T_{\rm fine}=T$.
Hence, the hierarchical design limits coarse-grained computation
to the coarsened graph and fine-grained optimization to matched
task-resource sets, reducing the overall matching and graph-processing
burden.

\section{Performance Evaluation\label{sectioniv}}
This section first presents the simulation setup and several comparison strategies, followed by evaluations addressing three key questions.

Q1: How quickly does the proposed strategy converge compared with existing task offloading strategies?

Q2: How robust is the proposed strategy across different problem sizes, fairness objectives, and key parameter settings?

Q3: How resilient is the proposed strategy to different attack scenarios compared with existing strategies?

\subsection{Simulation Setup}
We simulate a 3D space of 1000×1000×100 m, consisting of multiple UAVs and UGVs that are randomly distributed. 

\subsubsection{Parameter Settings of Tasks and Resources} We consider three task graph types (Bull, Star, and Tadpole) of UAVs \cite{tasktype} as shown in Fig. \ref{fig:tasktype}, where interdependent subtasks require contact durations $w^{ii'}_m \in [0.1, 0.3]$ \cite{liwanggraph}. Apart from graph topologies, each subtask possesses individual completion goals, i.e., $t^{\text{max}}_{im} \in [0.2, 0.9]$ s, $e^{\text{max}}_{im} \in [7, 14]$ kJ, $\theta^{\text{max}}_{im} = \{0.2, 0.5, 0.8\}$. New tasks are generated every 20 time slots, each lasting 0.05\,s. Besides, the edge weight of the resource graph follows $w_{nn^{'}} \in [0.05, 0.06]$ in small problem sizes and $w_{nn^{'}} \in [0.01, 0.02]$ in large problem sizes \cite{liwanggraph}.
\begin{figure}[!ht]
\centering
\includegraphics[width=0.40\textwidth]{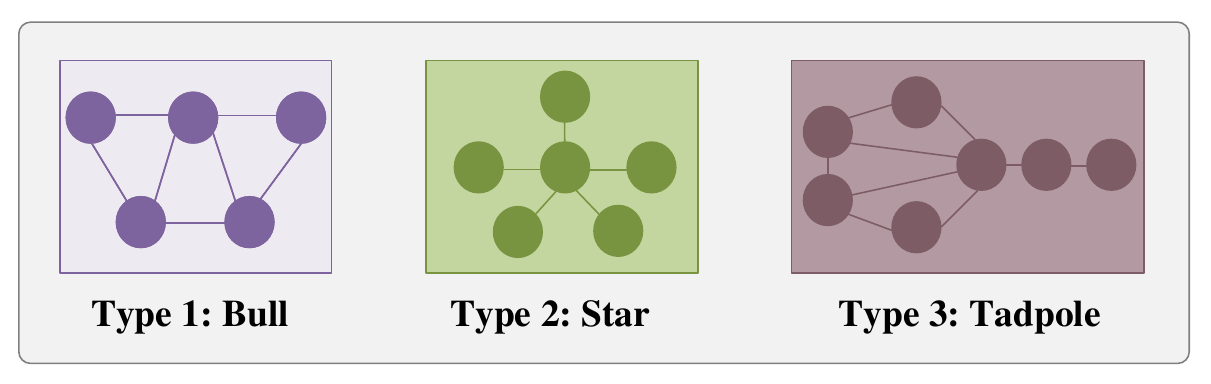}
\caption{Illustration of task graph types used in the simulations.}
\label{fig:tasktype}
\end{figure}

\subsubsection{Parameter Settings of Attacks} The simulations include ground Eves and malicious collaborators, accounting for 10\% and 20\% of the UGVs, respectively, following~\cite{chengoal2025}. On-Off attackers alternate between legitimate and malicious behaviors in 20-time-slot cycles~\cite{chengoal2025}. All proportions are rounded up when necessary.

The HTST-GNN uses two spatial aggregation layers with a learning rate of 0.001, structural dropout of 0, and temporal dropout of 0.5. The actor and critic each use 250 hidden units with learning rates of 0.001 and 0.01, respectively; $\gamma_0=0.9$, and Adam is used for optimization. These hyperparameters are empirically tuned with initial values guided by~\cite{Trustguard}. Other simulation parameters follow~\cite{communicationchannel,botaohyper} and are listed in Table~\ref{tab2}. All weighting parameters are set equally, i.e.,
$\varsigma_{1-6}=\xi_{1-3}=\frac{1}{3}$ and
$\epsilon_{1,2}=\zeta_{1,2}=0.5$.
\begin{table}[!ht]
    \renewcommand{\arraystretch}{1}
    \caption{Major Simulation Parameters}
    \label{tab2}
    \centering
    \begin{tabular}{>{\centering\arraybackslash}m{4.2cm}>{\centering\arraybackslash}m{3.8cm}}
        \hline\hline
        \textbf{Parameters} & \textbf{Value} \\
        \hline
        Subtask data size $D_{im}$ & [200, 1000] Kb \\
        Subtask processing density $\rho_{im}$ & \{500,600,700,800\} cycles/bit\\
        VM CPU clock frequency $f_{jn}$ & \{0.5, 0.8, 1\} GHz\\
        Jamming signal transmit power $p_{n}^{\text{jam}}$ & [50, 200] mW\\
        A2G channel bandwidth $W$ & [6, 8] MHz \\
        Transmission power $p_{m}$ & \{100,150,200\} mW \\
       Reference channel gain $g_1$ & -30 dB \\
       path loss exponent $\eta$ & 3 \\
       Maximum estimation error $\Delta g_{1,2}$ & $10^{-8}$ \\
       Noise power $\sigma ^2$ & -100 dBm \\
       Data exchange time threshold $\mu$ & [0.9, 1)\\
        \hline\hline
    \end{tabular}
\end{table}

\subsection{Comparison Strategies}
We compare the proposed strategy with representative baselines for dynamic collaborator selection. Each baseline removes or simplifies one or more core components to isolate their contributions.

\textit{1) \textbf{GDTrust} \cite{chengoal2025}}: A strategy that selects collaborators according to weighted trust scores updated upon task completion.

\textit{2) \textbf{TrustGNN} \cite{trustgnn}}: A GNN-based collaborator selection strategy using the same spatial architecture as ours but without temporal feature aggregation.

\textit{3) \textbf{SARSA} \cite{Sarsa}}: An RL strategy that maximizes long-term reward without considering trust and fairness.

\textit{4) \textbf{FairPG} \cite{fairmontecarlo}}: A fairness-aware policy-gradient strategy using Monte Carlo sampling.

\subsection{Convergence Analysis (Q1)}
\begin{figure*}[!ht]
\centering  
\subfigure[]{\includegraphics[width=0.31\textwidth]{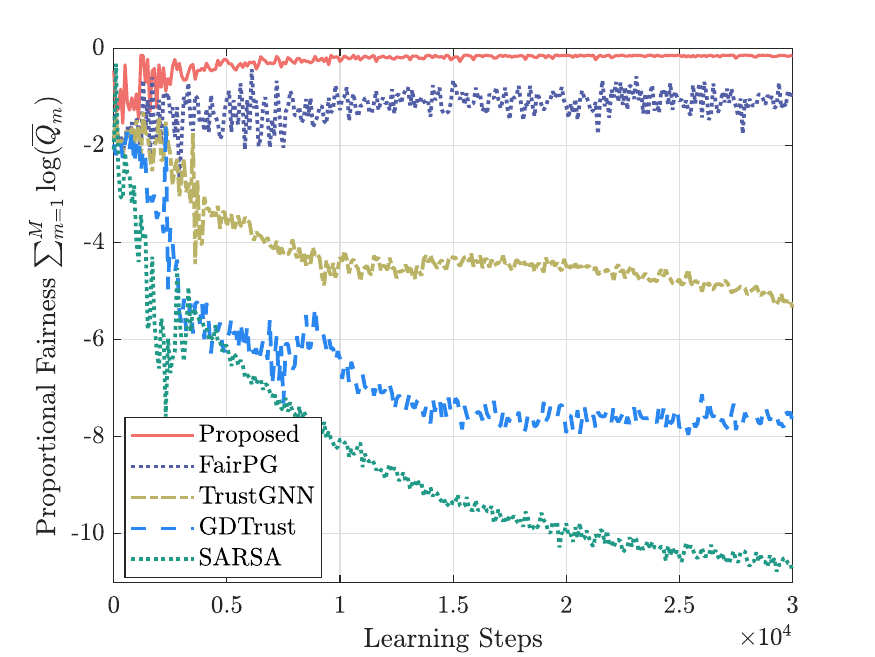}\label{fig:convergence1}}
\subfigure[]{\includegraphics[width=0.31\textwidth]{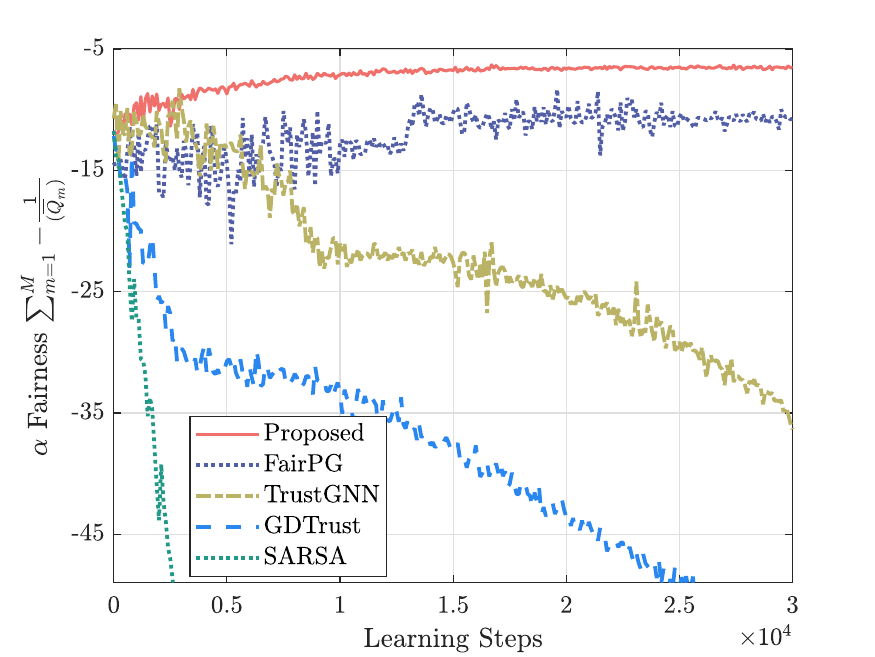}\label{fig:convergence2}}
\subfigure[]{\includegraphics[width=0.31\textwidth]{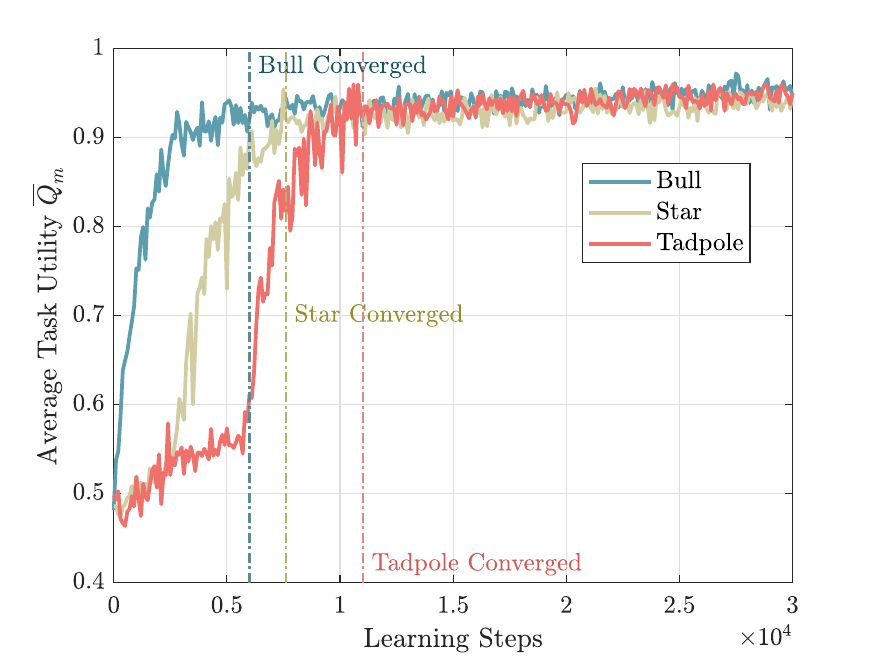}\label{fig:convergence3}}
\caption{Convergence comparison: (a) proportional fairness utility; (b) $\alpha$-fairness utility; (c) average task utility across task graph types.}
\label{fig:convergence}
\end{figure*}

\begin{figure*}[!ht]
\centering  
\subfigure[]{\includegraphics[width=0.39\textwidth]{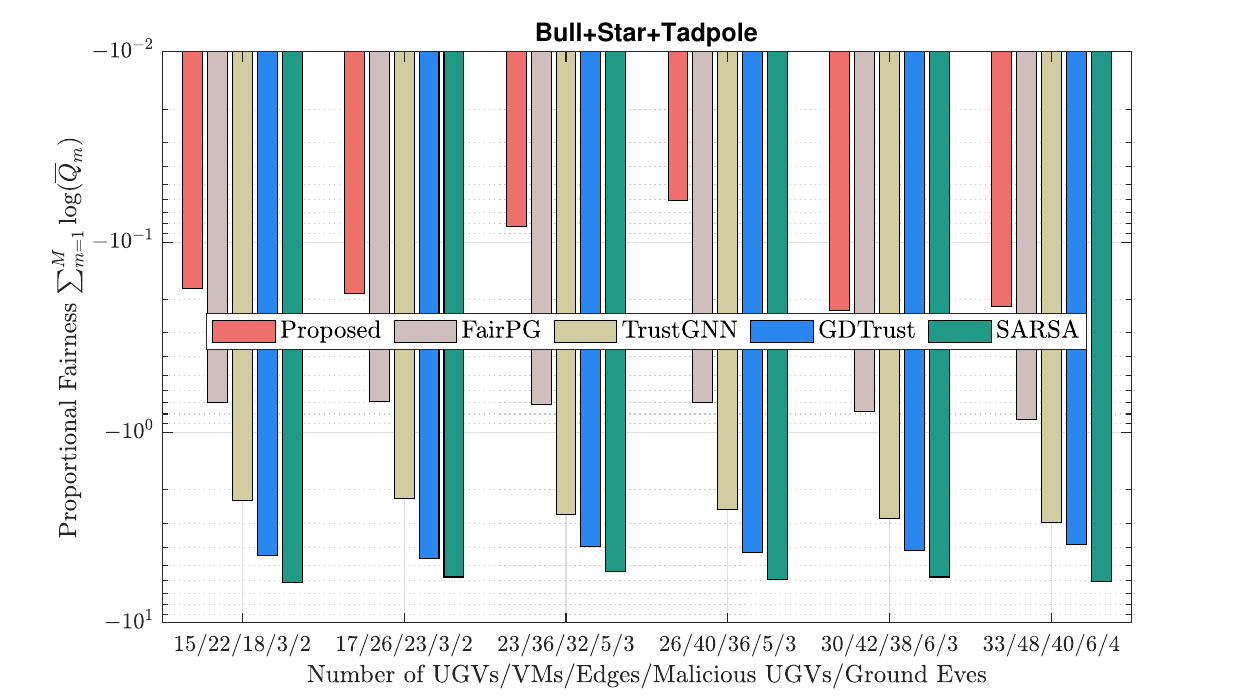}\label{fig:pictures4a}}
\subfigure[]{\includegraphics[width=0.39\textwidth]{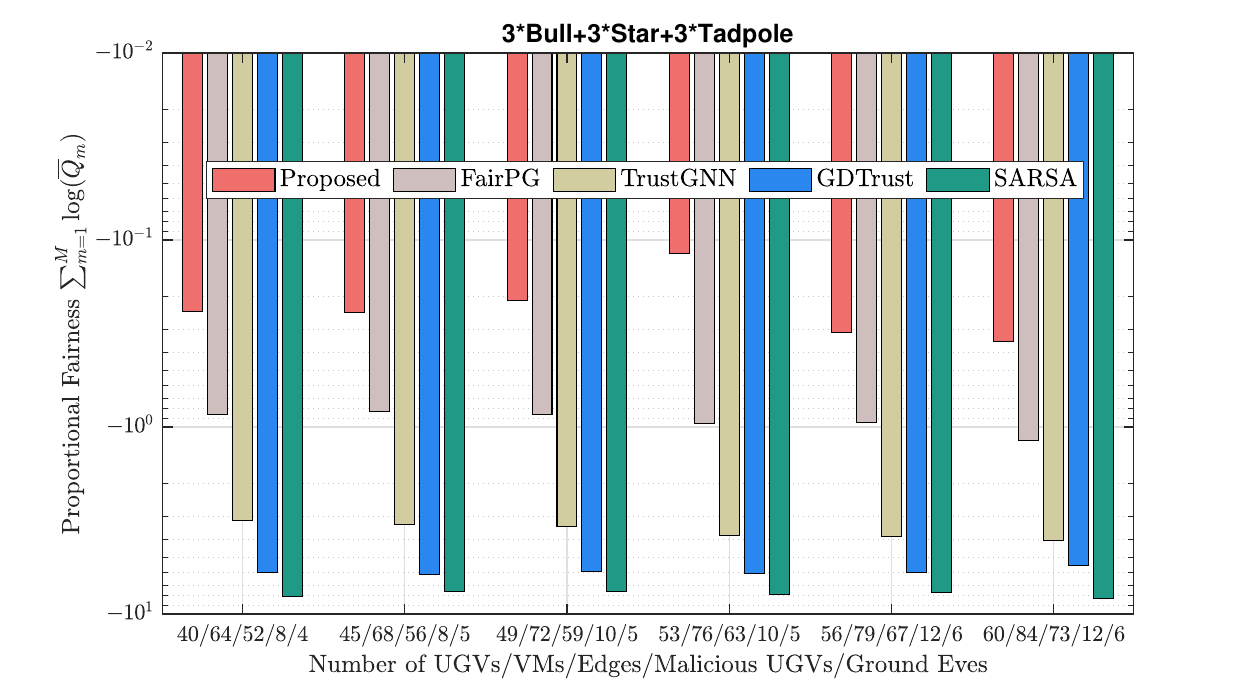}\label{fig:pictures4b}}
\subfigure[]{\includegraphics[width=0.39\textwidth]{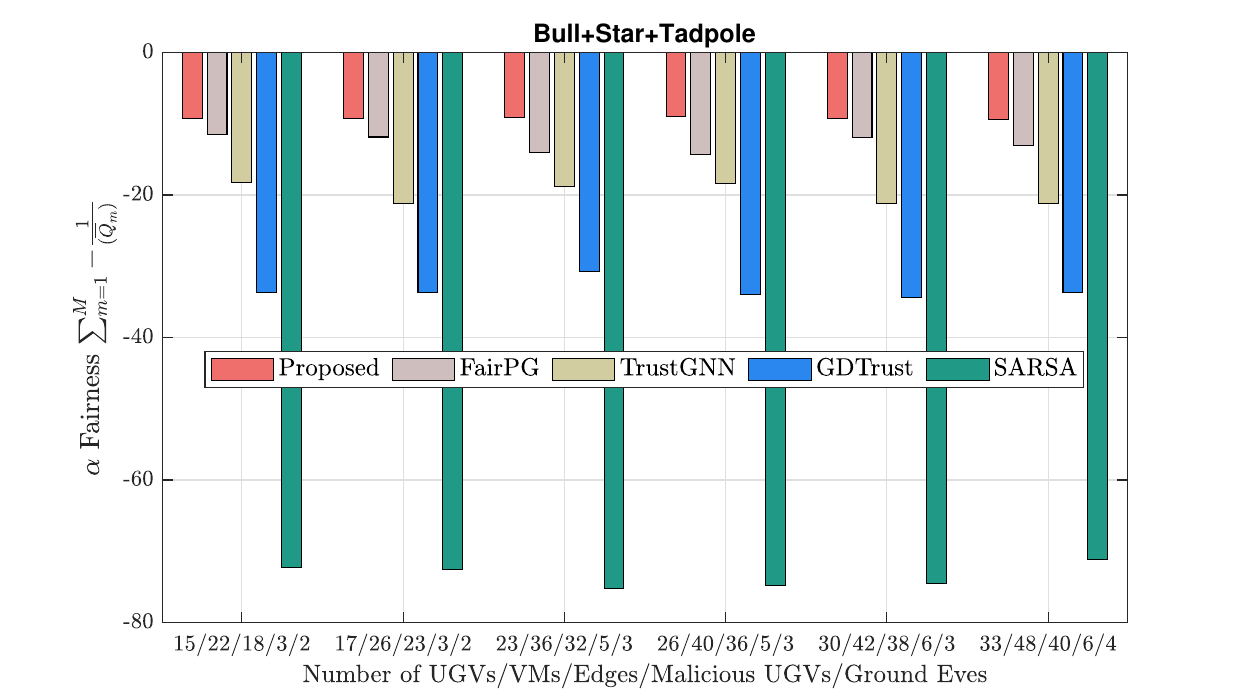}\label{fig:pictures4c}}
\subfigure[]{\includegraphics[width=0.39\textwidth]{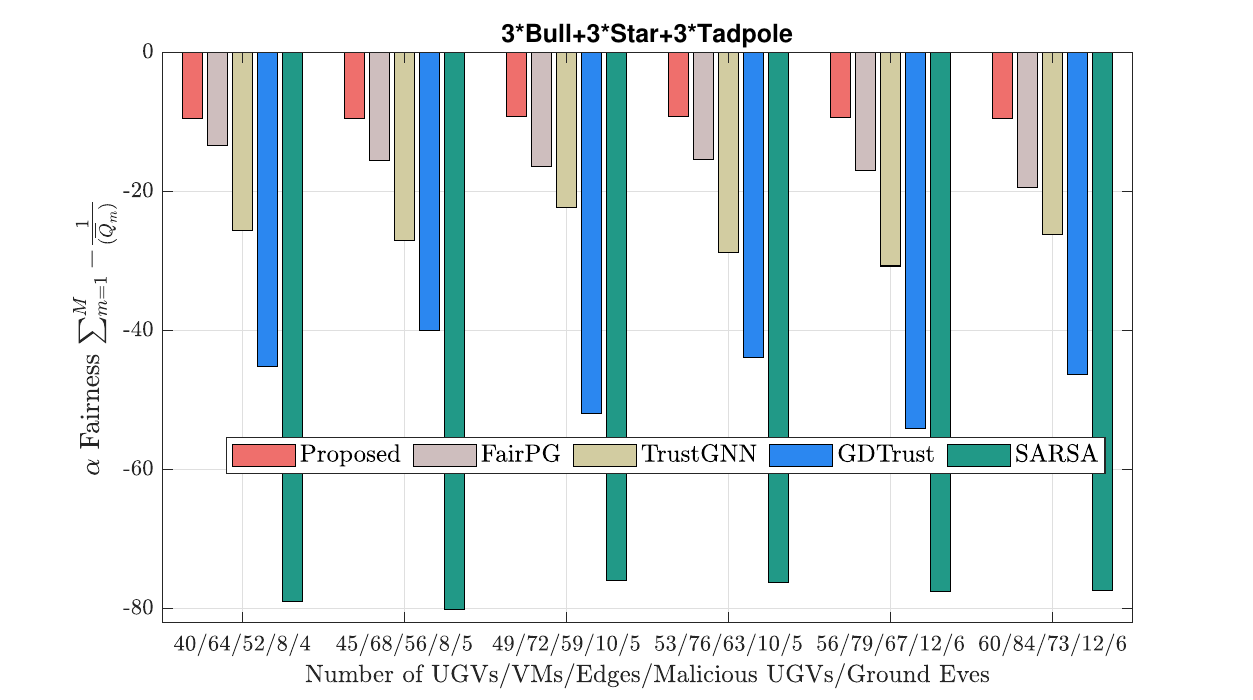}\label{fig:pictures4d}}
\caption{Fairness utility across problem sizes: (a)--(b) proportional fairness; (c)--(d) $\alpha$-fairness under small and large settings.}
\label{fig:pictures4}
\end{figure*}
Fig.~\ref{fig:convergence} evaluates convergence over 30,000 training steps with six task requestors, comprising two Bull, two Star, and two Tadpole tasks. For $\alpha=1$, (\ref{eqproblem}) reduces to proportional fairness optimization, $\sum_{m=1}^{M}\log(\overline{Q}_m)$. As shown in Fig.~\ref{fig:convergence1}, the proposed strategy converges within approximately 10,000 steps and stabilizes around $-0.2$. TrustGNN and GDTrust converge more slowly due to limited temporal trust modeling, while SARSA performs poorly without fairness consideration. FairPG also converges more slowly because of its Monte Carlo-based policy-gradient estimation.

We further evaluate the effectiveness of the proposed strategy to optimize other fairness objectives. For $\alpha=2$, the objective becomes
$\sum_{m=1}^{M}-1/\overline{Q}_m$, which penalizes low-utility tasks more heavily. Fig.~\ref{fig:convergence2} exhibits a similar trend: the proposed strategy achieves the fastest and most stable convergence, whereas FairPG converges more slowly and the remaining baselines achieve lower utility.

Fig.~\ref{fig:convergence3} further shows convergence across different task graph types. Bull tasks converge faster because of higher connectivity, whereas Star and Tadpole tasks initially exhibit lower or more fluctuating utility; nevertheless, their utilities gradually converge under the fairness-aware policy.

\begin{figure*}[!ht]
\centering  
\subfigure[]{\includegraphics[width=0.41\textwidth]{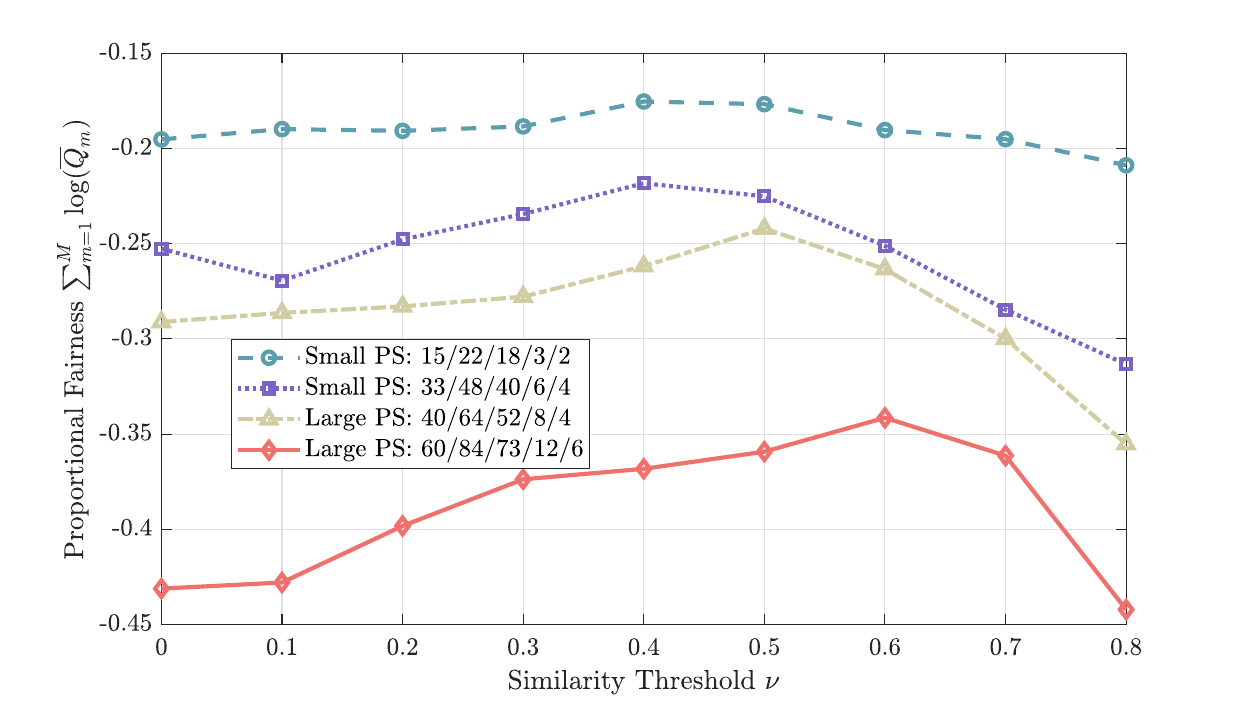}\label{fig:pictures5a}}
\subfigure[]{\includegraphics[width=0.41\textwidth]{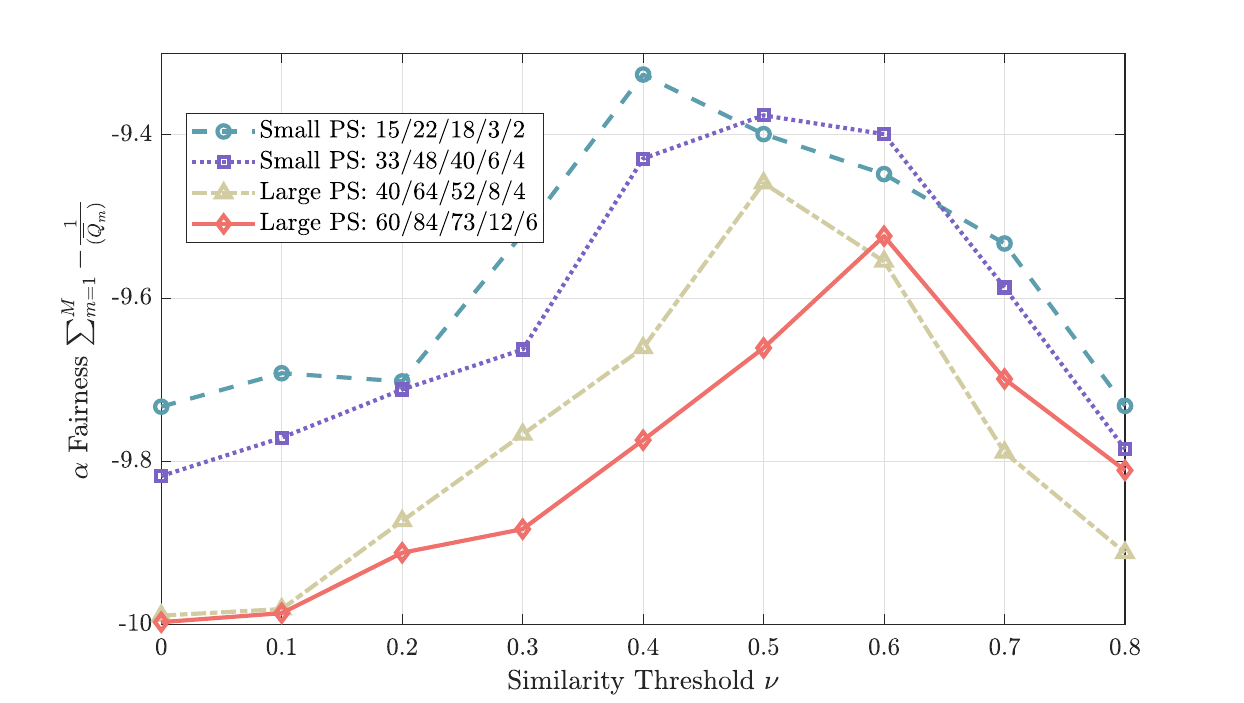}\label{fig:pictures5b}}
\caption{Impact of similarity threshold $\nu$ under different problem sizes: (a) proportional fairness utility; (b) $\alpha$-fairness utility.}
\label{fig:pictures5}
\end{figure*}

\subsection{Robustness and Parameter Sensitivity Analysis (Q2)}
We evaluate robustness across problem sizes, similarity thresholds, and fairness objectives. Task completion goals are randomly generated, with results averaged over 10,000 and 15,000 steps for proportional and $\alpha$-fairness utilities, respectively.

Fig.~\ref{fig:pictures4} shows that the proposed strategy consistently achieves the highest fairness utility across both small and large problem sizes and under both fairness objectives. TrustGNN, GDTrust, and SARSA perform worse because they do not jointly account for long-term fairness, dynamic trust, and hierarchical matching. Although FairPG explicitly optimizes fairness, its Monte Carlo-based policy gradient converges more slowly. These results demonstrate the scalability and robustness of the proposed hierarchical framework.

Fig.~\ref{fig:pictures5} evaluates the impact of the similarity threshold $\nu$ under different fairness objectives and problem sizes. Thresholds around 0.4--0.6 consistently achieve strong performance: smaller values retain more unrelated graph information and increase training overhead, whereas larger values prune relevant task relationships and reduce utility. We therefore set $\nu=0.5$ to balance prediction accuracy and computational efficiency.

\subsection{Attack Resistance Analysis (Q3)}
We evaluate blackhole, on-off, and eavesdropping attacks using 60 UGVs, 84 VMs, and 73 edges under the 3*Bull+3*Star+3*Tadpole setting. Fairness utilities are averaged over 10,000 and 15,000 steps for proportional fairness (PF) utility and $\alpha$-fairness (AF) utility, respectively. Attack resistance is measured by the fraction of malicious UGVs successfully avoided, while secrecy rate quantifies secure task-data transmission in the presence of Eves.
\begin{figure*}[!ht]
\centering  
\subfigure[]{\includegraphics[width=0.31\textwidth]{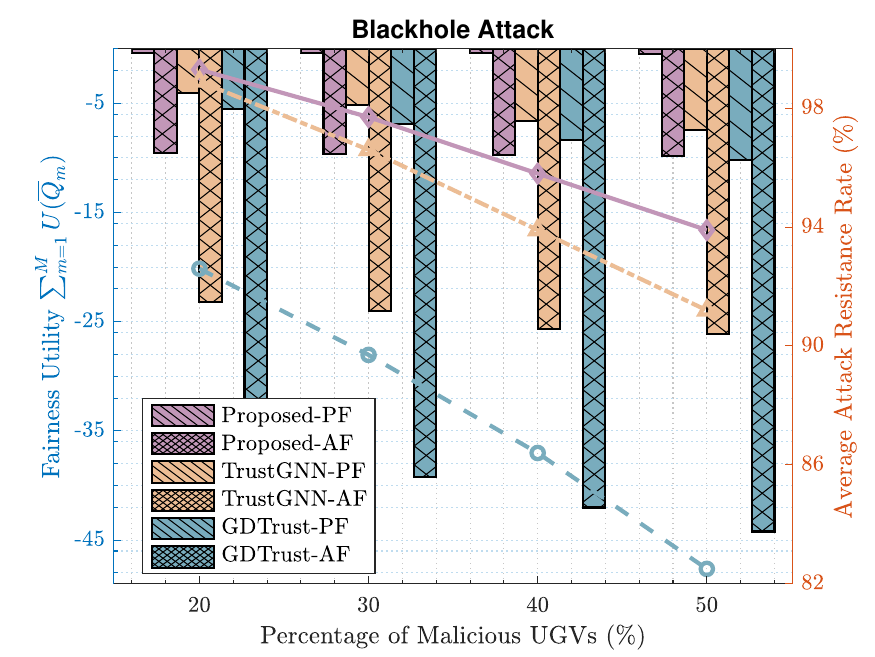}\label{fig:blackhole}}
\subfigure[]{\includegraphics[width=0.31\textwidth]{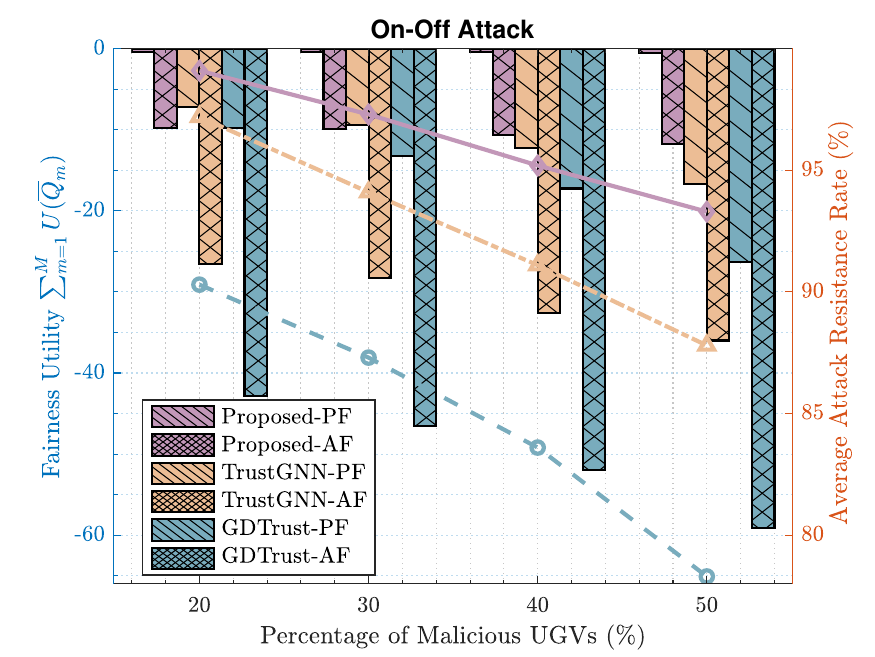}\label{fig:onoffattack}}
\subfigure[]{\includegraphics[width=0.31\textwidth]{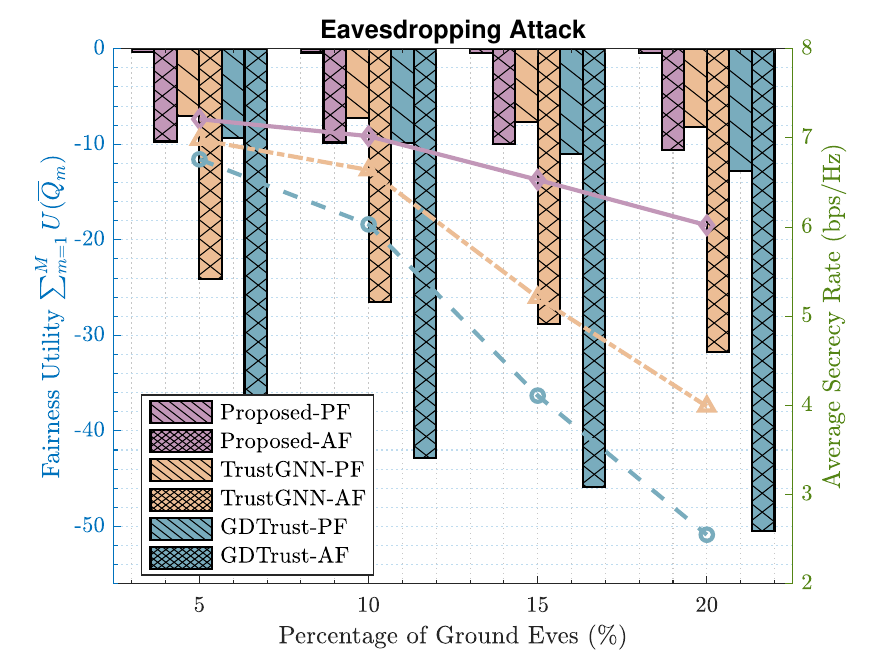}\label{fig:evesattack}}
\caption{Attack robustness under (a) blackhole, (b) on-off, and (c) eavesdropping attacks. In each subfigure, bars represent the proportional and \(\alpha\)-fairness utility values achieved by different strategies, while lines indicate average attack resistance and secrecy rates.}
\label{fig:attackanalysis}
\end{figure*}

For blackhole and on-off attacks, the malicious-UGV ratio is increased to 50\%. As shown in Figs.~\ref{fig:blackhole} and~\ref{fig:onoffattack}, on-off attacks are more disruptive because intermittent malicious behavior reduces detectability and allows attackers to maintain deceptive trust values. Nevertheless, the proposed strategy degrades more gracefully than TrustGNN and GDTrust and maintains approximately 93\% attack resistance even with 50\% on-off attackers.

Fig.~\ref{fig:evesattack} shows that secrecy rates decrease as the number of ground Eves increases from 3 to 12. The proposed strategy nevertheless maintains higher fairness utility and stronger resilience than the baselines because its trust-aware adaptive policy reduces the impact of unreliable collaborators and compromised links.

\section{Conclusion\label{sectionv}}
This paper has proposed a fairness-aware concurrent task orchestration strategy for trusted collaborator group formation and task allocation in resource-constrained collaborative systems. A unified task utility model captures the value–risk tradeoff of task completion, based on which concurrent task orchestration is formulated as a long-term average utility maximization problem under $\alpha$-fairness. To solve this problem, a hierarchical coarse-to-fine matching framework integrates task-specific trust prediction via spatial-temporal GNNs with fairness-aware DRL. The predicted trust values provide abstracted state representations for matching decisions, while fairness-aware rewards balance long-term task utilities under resource contention. Simulation results demonstrate improved task completion performance over representative baselines and consistent robustness under varying problem sizes, fairness objectives, and attack conditions. In particular, the proposed strategy improves concurrent task completion performance by 22\% over representative baselines and maintains an average attack resistance rate of approximately 93\% even when 50\% of UGVs launch on-off attacks.

\bibliographystyle{IEEEtran}
\bibliography{IEEEabrv,reference}

\end{document}